\documentclass[aps,twocolumn,superscriptaddress]{revtex4-1}

\usepackage{graphicx}
\usepackage{dcolumn}
\usepackage{bm}
\usepackage{braket}
\usepackage{dsfont}
\usepackage{xcolor}

\usepackage{hyperref}
\usepackage{amsmath}
\usepackage{amsfonts}
\usepackage{amsthm}
\usepackage{enumitem, hyperref}
   \newtheorem{theorem}{Theorem}
   \newtheorem{thm1}{Theorem}
   \newtheorem{Assumption}{Assumption}
   
   \newtheorem{lemma}{Lemma}
   \newtheorem{proposition}{Proposition}
   \newtheorem{definition}{Definition}

\begin{document}


\title{Reconstruction of Quantum Particle Statistics: Bosons, Fermions, and Transtatistics}

\author{Nicolás Medina S\'anchez}
\email{nicolas.medina.sanchez@univie.ac.at}
\affiliation{University of Vienna, Faculty of Physics, Vienna Center for Quantum Science and Technology (VCQ), Boltzmanngasse 5, 1090 Vienna, Austria}
\affiliation{University of Vienna, Vienna Doctoral School in Physics, Boltzmanngasse 5, 1090 Vienna, Austria}
\author{Borivoje Daki\'c}
\email{borivoje.dakic@univie.ac.at}
\affiliation{University of Vienna, Faculty of Physics, Vienna Center for Quantum Science and Technology (VCQ), Boltzmanngasse 5, 1090 Vienna, Austria}
\affiliation{Institute for Quantum Optics and Quantum Information (IQOQI), Austrian Academy of Sciences, Boltzmanngasse 3, 1090 Vienna, Austria.}

\date{\today}

\begin{abstract}
Identical quantum particles exhibit only two types of statistics: bosonic and fermionic. Theoretically, this restriction is commonly established through the symmetrization postulate or (anti)commutation constraints imposed on the algebra of creation and annihilation operators. The physical motivation for these axioms remains poorly understood, leading to various generalizations by modifying the mathematical formalism in somewhat arbitrary ways. In this work, we take an opposing route and classify quantum particle statistics based on operationally well-motivated assumptions. Specifically, we consider that a) the standard (complex) unitary dynamics defines the set of single-particle transformations, and b) phase transformations act locally in the space of multi-particle systems. We develop a complete characterization, which includes bosons and fermions as basic statistics with minimal symmetry. Interestingly, we have discovered whole families of novel statistics (dubbed transtatistics) accompanied by hidden symmetries, generic degeneracy of ground states, and spontaneous symmetry breaking-- effects that are (typically) absent in ordinary statistics.
\end{abstract}

\maketitle

\section{Introduction}\label{sec:level1}
The concept of identical particles was introduced by 
Gibbs in 1902~\cite{gibbs} as an alternative to solve the problem related to the extensitivity of entropy, the so-called \emph{Gibbs paradox}. According to Gibbs, a system consists of identical particles if its physical magnitudes are invariant under any permutation of its elements. Bose has put forward this idea in quantum mechanics in his derivation of Planck's law of blackbody radiation~\cite{Bose1924}. This was further developed by Dirac \cite{dirac} and Heisenberg \cite{heisenberg}, who formulated the well-known \emph{symmetrization postulate}: physical states must be symmetric in such a way that the exchange of particles does not give any observable effect. Put in the standard language of wavefunctions, if the state of, say, two particles is given by $\psi(x_1,x_2)$, then
\begin{equation}\label{symm. postulate}
\psi(x_2,x_1)=e^{i\varphi}\psi(x_1,x_2).
\end{equation}
Applying the particle swap twice trivially reveals $e^{i\varphi}=\pm1$. This is the origin of two types of particle statistics: \emph{bosons} (symmetric) and \emph{fermions} (antisymmetric). 

Another approach to explain the origin of quantum statistics is the \emph{topological} argument~\cite{myrheim,dewitt}. Namely, the exchange symmetry is directly related to the continuous movement of particles in a physical (configuration) space, which implies that only bosonic and fermionic phases are allowed, given that the number of spatial dimensions is three or greater. In lower dimensions, one gets fractional phases and anyonic statistics~\cite{Wilczek82}. 

The third common way of addressing the question of particle statistics is to take the \emph{algebraic (field)} approach~\cite{weinberg95}, i.e., by postulating the set of canonical relations
\begin{equation}\label{}
    [a_i,a_j^{\dagger}]_{\pm}=\delta_{ij}\openone,
\end{equation}
where $\pm$ stands for (anti)commutator of operators (for fermions and bosons, respectively). Starting with an assumption of a unique vacuum state, one can build the multi-particle state space (Fock space) for two types of particle statistics.

While these approaches agree at the level of ordinary statistics (bosons and fermions), all of them have been criticized for their \emph{ad hoc} nature~\cite{Messiah64,mirman1973,vanenk2019}. This leaves the door open for various generalizations, many of which resort to somewhat arbitrary assumptions added to the quantum formalism. Earliest work along these lines dates back to Gentile and his attempt to interpolate between two statistics~\cite{Gentile40}, and since then, we have seen dozens of generalized and exotic statistics, such as parastatistics~\cite{Green53}, quons and intermediate statistics~\cite{quons,Greenberg99,LAVAGNO10,Fivel90}, infinite statistics~\cite{Greenberg90,Medvedev97}, generalizations of fractal and topology-dependent statistics~\cite{CHEN_1996,Polychronakos99,cattani09,Surya_2004}, ewkons~\cite{Hoyuelos16} modifications of statistics due to quantum gravity~\cite{SWAIN_2008,balachandran_2001,baez06}, non-commutative geometry~\cite{arzano09} and others~\cite{Maslov2009,trifonov09,Bagarello_2011,NIVEN09}.
\subsection{Operational approach and particle statistics}
So far, exotic statistics have never been observed in nature. This situation can be interpreted at least in two ways: we need more sophisticated and precise experiments, or (some) generalizations are in collision with basic laws of physics (believed to hold universally). An excellent example of the latter point is a question of the parity superselection rule (PSR) for fermions derived from the impossibility
of discriminating a $2\pi$-rotation from the identity in three-dimensional space~\cite{Wightman1995}. One may wonder how to apply this reasoning in a more abstract scenario, such as fermions occupying some discrete degrees of freedom (e.g., energy) where no notion of rotation (\emph{a priori}) exists. An elegant study was provided in a recent work~\cite{johansson2016} based on techniques from quantum information, showing that a PSR violation would allow for superluminal communication. Thus, the parity superselection rule can be derived from a more basic law, i.e., the \emph{no signaling} principle~\cite{ghirardi2013}. Such an approach to physical theories (from physical laws to mathematical formalism) resembles Einstein's original presentation of special relativity. In that case, a concise set of physical postulates, namely the covariance of physical laws and the constancy of the speed of light in all frames of reference, paved the way for the formalism of Lorentz transformations.
In the realm of quantum foundations, the application of this methodology was particularly successful. With the pioneering work of Hardy~\cite{hardy2001}, the field of \emph{operational reconstructions} of quantum theory~\cite{hardy2001, dakic11, chiribella2011, Masanes_2011, Dakic2016, Hohn17} was established where one recovers the abstract machinery of Hilbert spaces starting from a set of information-theoretic axioms. Considering the significance of identical particles in quantum information processing (such as in linear optical quantum computing~\cite{knill2001scheme}), it becomes evident that utilizing this operational approach holds significant potential to derive particle statistics based on physically grounded assumptions.
Rather than exploring possible modifications of the existing formalism, a more constructive approach may begin by defining a typical quantum experiment and addressing straightforward physical questions. For example, how do we define particle (in)distinguishability from an experimental standpoint? Is it possible to establish a clear operational differentiation between various types of identical particles, and if so, how do we characterize the corresponding mathematical formalism? Our work can be understood as an attempt to answer these questions. Along these lines, promising research studies appeared in the context of the symmetrization postulate~\cite{Goyal_2019}, anyonic statistics~\cite{neori16}, quantum field theory~\cite{DAriano14,eon2022} and identical particles in the framework of generalized probabilistic theories~\cite{D_Ariano_2014,dahlsten2013}. 
\subsection{Reconstruction, mathematical foundations and summary of the results}

Following the instrumentalist approach of Hardy~\cite{hardy2001}, we study identical quantum particles in an operationally well-defined setup composed of laboratory primitives, such as preparations, transformations, and measurements (see Fig. \ref{setup}). Our starting point is a single quantum particle which we assume is an ordinary quantum particle described by standard formalism and unitary dynamics. This appears rather natural, as a single quantum particle is insensitive to statistics. We introduce a typical apparatus for a single-particle transformation described by a unitary channel on $d$-modes ($d\times d$ unitary matrix) and a set of $d$ detectors at the output. For such a fixed circuit,  we investigate the scenario with multiple identical particles at the input and analyze the probability of detecting them after the transformation. Detectors can register only particle numbers but cannot distinguish them; thus, 
 indistinguishability is built in from the beginning. As we shall see, the Fock-space structure will naturally arise as an ambient space for multi-particle states. Two central mathematical ingredients will thus figure prominently in our reconstruction of particle statistics: 
\begin{itemize}
    \item[1)] unitary group $U(d)$ describing single-particle transformations, and 
    \item[2)] the Fock space structure encompassing multi-particle states.
\end{itemize}

Paired with the locality assumption (i.e., phase transformations acting locally in Fock space), these two elements will determine how particles are organized in multiparticle states. Mathematically, the problem concerns the classification of representations of the $U(d)$ group in Fock space subjected to locality constraint. We found a one-to-one correspondence to the well-studied mathematical problem of characterizing completely-positive sequences~\cite{bump,Schoenberg,Davydov,integer}. This, in turn, provided us with a \emph{complete categorization of particle statistics} based on integral polynomials. To be more precise, a list of integers
\begin{equation}
 [q_0,q_1,\dots]_{\pm},   
\end{equation}
defines a type of particle statistics, provided that $Q_{\pm}(x)=\sum_{s}^{}(\mp1)^sq_sx^s$ are polynomials with all negative ($+$) or positive ($-$) roots. We coin the term \emph{transtatistics} for this generalized statistics. This is a natural generalization of ordinary statistics into two types: fermionic-like $[\dots]_{-}$ (\emph{transfermions}) and bosonic-like $[\dots]_{+}$ (\emph{transbosons}), and to the best of our knowledge, was not presented in the literature. Ordinary statistics is the simplest possibility (degree-one) $[1,1]_{\pm}$ with multiparticle Fock state being completely specified by irreducible representations (IR) of $U(d)$. On the other hand, general transtatistics requires additional quantum numbers to identify states of indistinguishable particles; thus, \emph{hidden symmetries}~~\cite{mcintosh1959accidental} and \emph{new degrees of freedom} emerge exclusively from these types of particles. We discuss further physical consequences by analyzing the thermodynamics of non-interacting gases. In doing so, we find an interesting inclusive degeneracy of ground-states followed by \emph{spontaneous symmetry breaking}~\cite{peierls1991spontaneously}, which (usually) does not exist in ordinary statistics. 

Symmetry is central to our reconstructions. In particular, the $U(d)$ symmetry of single-particle transformations is uniquely related to ordinary statistics and transtatistics. This also brings the main difference to other generalized statistics, which rely on different symmetries. Apart from the foundational relevance, our findings apply to quantum information and quantum many-body physics. Concretely speaking, transtatistics brings novel theoretical models for non-interacting identical particles. The latter is relevant for studying strongly-correlated quantum systems (see ~\cite{Verstraete21} and references therein), many of which are reducible to non-interacting models of indistinguishable particles~\cite{lieb61}. Therefore, one may find new integrable models among strongly interacting quantum systems reducible to our non-interacting model. On the quantum information side, quantum statistics is essential in complexity theory and intermediate quantum computing models, such as in boson sampling~\cite{aaronson11}. In this respect, our classification is relevant as it may lead to the discovery of new intermediate computational models. These points are only summarized here and will be discussed in more detail in the last section of the manuscript. 

\section{Operational setup for indistinguishable particles}\label{Operational setup}
\begin{figure}
    \centering
    \includegraphics[scale=0.6]{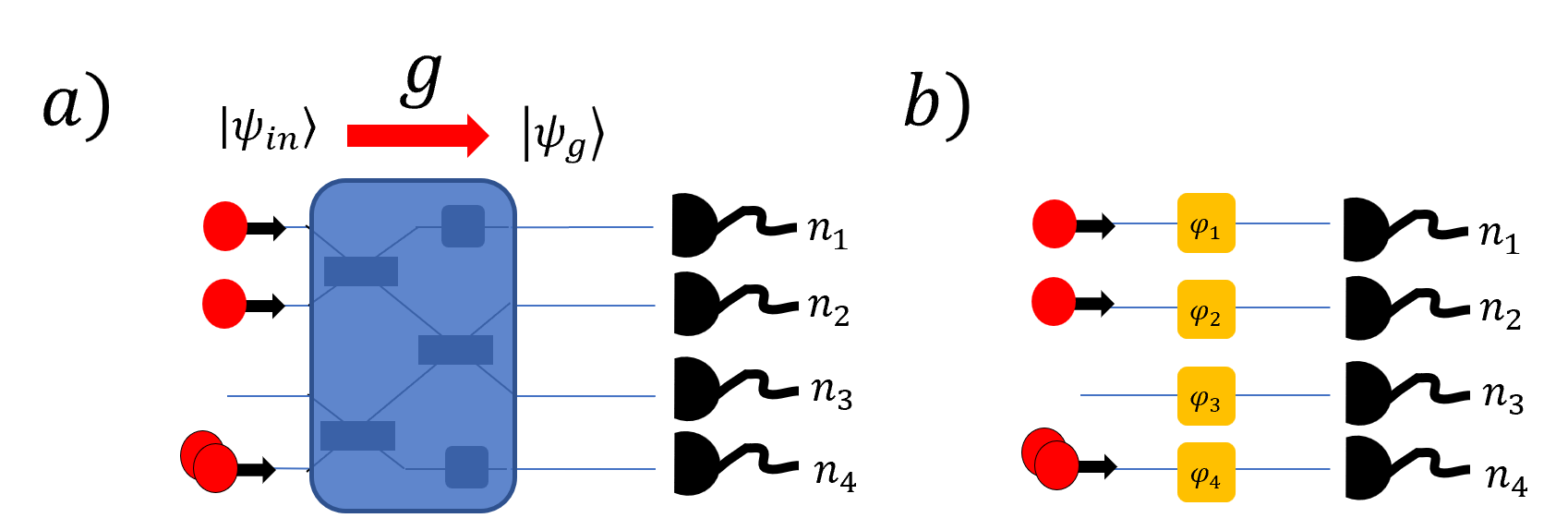}
    \caption{{\bf Operational setup}. a) Quantum circuit represented by $U(d)$-transformation for a single particle (example of $d=4$ is shown). For many particles injected into the setup, detectors can register their number only. This incorporates the notion of indistinguishability. b) Disconnected (independent) phase gates acting on particles locally, in individual modes.}
    \label{setup}
\end{figure}
The operational framework for indistinguishable particles is illustrated in Fig. \ref{setup}$a)$. The apparatus consists of $d$ modes into which particles can be injected, followed by a transformation $g$ and a set of $d$ detectors that register particles after the transformation. The transformation $g$ is fixed and independent of particle number at the input and particle-statistics type. One can think of this transformation as a quantum circuit composed of elementary gates, such as beam-splitters and phase shifters used in quantum linear optics to produce a general unitary transformation $g\in U(d)$ on $d$ modes~\cite{Pan12}, where $U(d)$ is the set of $d\times d$ unitary matrices. As long as just one particle is injected into the setup, e.g., in mode $i$, the $j$th detectors will register the particle with the probability $p_j=|g_{ji}|^2$ with $g_{ji}$ being the matrix element of $g$. In other words, $g$ represents standard complex unitary dynamics of a single quantum particle with $d$ levels (modes). The critical question to be answered is what will happen if more than one particle is injected into such apparatus? To formalize the situation, there are three points to be addressed in the first place:
\begin{itemize}
    \item [i)] We shall determine the ambient Hilbert space describing the multi-particle system, 
    \item [ii)] we have to find the corresponding representation of transformations (as defined by the $U(d)$ group) in such a space, and
    \item [iii)] finally, determine the Born rule to calculate probabilities of detection events.  
\end{itemize} 
To identify the Hilbert space of many particles, we use the fact that particles are indistinguishable, i.e., detectors can register only particle numbers (how many particles land in a particular detector without distinguishing them). Thus the overall measurement outcome is described by a set of numbers $(n_1,n_2,\dots,n_d)$, with $n_k=0,1,2,\dots$. This outcome fully specifies the physical configuration; thus, we associate to it the measurement vector $\ket{n_1,\dots n_d}$ such that the Born rule gives detection probabilities  
\begin{equation}\label{detection probs}
    p_{n_1,\dots,n_d}=|\braket{n_1,\dots,n_d|\psi_g}|^2,
\end{equation}
where $\ket{\psi_g}$ is the state of the system after the transformation $g$. From here, we directly see that $\ket{\psi_g}\in\mathcal{F}_d$ resides in a Fock space defined as a span over number states, i.e.,  
\begin{equation}\label{Fock space}
\mathcal{F}_d=\mathrm{span}\{\ket{n_1,n_2,...,n_d}~|~n_k=0,1,2,\dots p\}.    
\end{equation}
Here $span$ denotes the complex linear span (hull) of basis vectors. Since outcomes $(n_1,n_2,\dots,n_d)$ are perfectly distinguishable, vectors $\ket{n_1,n_2,...,n_d}$ form an orthonormal set. We introduced the possibility of there being a maximal occupation number $p\in\mathbb{N}$, which is the \emph{generalized Pauli exclusion principle}. As we shall see, $p=1$ will correspond to fermionic statistics, while bosons are associated with the case $p=+\infty$. At this stage, $p$ is characteristic of statistics and is kept as an integer parameter (possibly infinite). Note that the Fock space in \eqref{Fock space} shall not be \emph{a priori} identified with the standard (textbook) Fock space constructed as a direct sum of particle sectors. Our Fock space is an ambient Hilbert space for multi-particle states naturally emerging from operational considerations and the measurement postulate defined in \eqref{detection probs}. Note also that the Fock space in \eqref{Fock space} is of the tensor product form, i.e., $\mathcal{F}_d=\mathcal{F}_1^{\otimes d}$.

Now, we shall find an appropriate unitary representation of $g\in U(d)$ in the ambient space $\mathcal{F}_d$, i.e. $\Delta_d: U(d)\mapsto GL(\mathcal{F})$ such that 
\begin{equation}\label{U(d) action}
    \ket{\psi_g}=\Delta_d(g)\ket{\psi_{in}},
\end{equation}
with $\Delta_d(g)$ being unitary representation and $\ket{\psi_{in}}\in\mathcal{F}_d$ is some input state to the circuit in Fig. \ref{setup}$a)$. For example, $\ket{\psi_{in}}=\ket{1,1,0,\dots,0}$ represents the input state of two particles injected in mode $1$ and $2$. In general, $\ket{\psi_{in}}$ may involve the superposition of number states. 
Representation $\Delta_d$ is reducible in general, and the \emph{group character} completely determines its decomposition into irreducible (IR) sectors~\cite{fulton}, that is, a function defined over the elements of the group
\begin{equation}\label{group character}
    \chi_d(g) =\mathrm{Tr}(\Delta_d(g)),\quad\forall g\in U(d).
\end{equation}  
As we shall see, the irreducible decomposition of Fock space \eqref{Fock space} will be in one-to-one correspondence to the type of particle statistics. So, the group character will be our main object of interest.
\subsection{Locality assumption}
To evaluate character on $U(d)$ group, recall that any unitary matrix can be diagonalized, i.e., $g=StS^{\dagger}$, with $t=\mathrm{diag}[x_1,\dots,x_d]\in T_d=U(1)\times\dots\times U(1)$ being an element of the maximal torus (also known as the phase group) with $x_k=e^{i\theta_k}\in U(1)$. Therefore, the character of $U(d)$ is entirely specified by the character evaluated on $T_d$, that is, $\chi_d(StS^{\dagger})=\mathrm{Tr}\Delta_d(StS^{\dagger})=\mathrm{Tr}\Delta_d(t)=\chi_d(t)$ (i.e., class function), thus it effectively becomes a function of phase variables, i.e., $\chi_d(\Vec{x})=\chi_d(x_1,\dots,x_d)$. 

Consider the case of a single-mode ($d=1$) with the Fock space $\mathcal{F}_1=\mathrm{span}\{\ket{n}|~n=0,1,\dots,p\}$ on which the group $U(1)$ acts with representation $\Delta_1(x)$, with $x=e^{i\theta}$. We can think of $\Delta_1(x)$ representing a simple device providing a phase shift to the state of a single particle placed in a mode. We can now consider the collection of $d$ such devices disconnected from each other and operating independently in separate modes, as illustrated in Fig. \ref{setup}$b)$. These transformations form the phase group $T_d$ acting in the entire Fock space, and given their operational independence, it appears natural to assume the following.
\begin{Assumption}[{\bf Locality}]\label{A0}\emph{The action of the phase group $T_d$ in Fock space is local, i.e.,
\begin{equation}\label{Torus action}
     \Delta_d(\vec{x})=\Delta_1(x_1)\otimes\dots\otimes\Delta_1(x_d), 
    \end{equation}
for $\Vec{x}\in T_d$.}
  \end{Assumption}
By taking the trace of the last equation, one gets
\begin{equation}\label{character}
    \chi_d(\vec{x})=\prod_{k=1}^d\chi_1(x_k),
\end{equation}
with $\chi_1(x)=\mathrm{Tr}(\Delta_1(x))$ being the single-mode character. One can also go in the reversed direction, i.e., starting with the character factorization in  \eqref{character}, we may derive the tensor factorization in \eqref{Torus action}, which follows from general character theory~\cite{fulton}.

Assumption \ref{A0} is our central assumption. We see that the single-mode character $\chi_1$, a function of a single variable, entirely specifies the character of the whole $U(d)$. The problem then simplifies, and our goal is to determine the most general form of $\chi_1(x)$ such that $\chi_d(\vec{x})$ in \eqref{character} is a valid character of $U(d)$. 
\subsection{Generalized number operator and conserved quantities}\label{Generalized N number}
What follows from Assumption \ref{A0} and factorization given in \eqref{character} is that the single-mode character $\chi_1(x)$ completely specifies the character of the whole $U(d)$ and consequently determines the decomposition of Fock space into IR sectors. Note that the action of the single-mode phase transformation $x=e^{i\theta}\in U(1)$ can be seen as an instance of the Hamiltonian evolution. Thus we can write $\theta=\epsilon t/\hbar$, where $\epsilon$ is the single-particle energy associated with this mode. 
With this, the representation of the phase transformation becomes $\Delta_1(e^{i/\hbar\epsilon t})=e^{i/\hbar \hat{H}t}$, where $\hat{H}$ is the single-mode Hamiltonian (generator of phase). From the invariance under $(2\pi)$-rotations, i.e., $e^{i(\theta+2\pi)}=e^{i\theta}$, we conclude that all eigenvalues of $\hat{H}$ are integer multiples of $\epsilon$, that is, $\hat{H}=\epsilon \tilde{N}$ with $\Tilde{N}$ being the operator with integer eigenvalues. This defines the \emph{generalized number operator} or \emph{excitation operator} $\Tilde{N}$. Here, we will consider only the case $\tilde{N}\geq0$;  the possibility of negative eigenvalues of $\Tilde{N}$, which would account for the most general generator of $U(1)$, will be discussed later in Section \ref{Negative occupations}. Without loss of generality, we can assume $U(1)$ action to be number preserving, thus 
\begin{equation}\label{number operator}
\tilde{N}=\sum_{n=0}^pf_n\ket{n}\bra{n}, 
\end{equation}
with $f_n$ being non-negative integers. Note that $\tilde{N}$ is in general different from the standard number operator $\hat{N}=\sum_{n=0}^pn\ket{n}\bra{n}$. The two will coincide only if $f_n=n$, and as we shall see, this happens only in the case of ordinary statistics. 

Finally, we can write the single-mode character $\chi_1(e^{i\theta})=\mathrm{Tr}(e^{i\theta\tilde{N}})$ as 
\begin{equation}\label{single-mode character}
    \chi_{1}(x)=\sum_{s=0}^{+\infty}a_sx^s=x^{f_0}+x^{f_1}+\dots+x^{f_p},
\end{equation}
with $a_s$ being a non-negative integer. Mathematically speaking, the formula above is the decomposition of $\chi_1$ into irreducible representations of $U(1)$. For fermions, we have $\chi_1^{(-)}(x)=1+x$, while for bosons $\chi_1^{(+)}(x)=1+x+x^2+\dots=\frac{1}{1-x}$.

For the case of $d$ modes, the action of the phase group in \eqref{Torus action} becomes
\begin{equation}\label{Torus action N}
     \Delta_d(\vec{x})=e^{i\theta_1\tilde{N}_1+\dots+i\theta_d\tilde{N}_d}, 
    \end{equation}
where $\Tilde{N}_k=\openone^{\otimes (k-1)}\otimes\tilde{N}\otimes\openone^{\otimes (d-k)}$ are generators of local phases. 
The vector $\vec{x}=\theta(1,1,\dots,1)^T\in T_d$ corresponds to the scalar $d\times d$ matrix $e^{i\theta}\openone_d$ commuting with all $U(d)$ matrices, thus the operator
\begin{equation}\label{Casimir}
    \Tilde{N}=\sum_{k=1}^d\Tilde{N}_k,
\end{equation}
is a conserved quantity (Casimir operator) and represents the total number of excitations. We can also write \eqref{Torus action N} as being generated by the following Hamiltonian
\begin{equation}\label{Hamiltonian}
    \hat{H}=\sum_{k=1}^d\epsilon_k\Tilde{N}_k,
\end{equation}
where $\theta_k=\epsilon_k t/\hbar$ and $\epsilon_k$ is are the single-particle energies.

\section{Particle statistics and its classification}
\subsection{On exchange symmetry}
In the $1$st quantization approach, particle statistics are classified via the exchange of particles and symmetrization postulate as given in equation \eqref{symm. postulate}. However, this method does not apply to the Fock-space approaches simply because there is no particle label (they are indistinguishable). A partial solution to this problem is to introduce \emph{permutation of modes operator}~\cite{stolt1970} 
\begin{equation}
    \Delta_d(\sigma)\ket{n_1,n_2,\dots,n_d}=\ket{n_{\sigma(1)},n_{\sigma(2)},\dots,n_{\sigma(d)}},
\end{equation}
for some permutation $\sigma\in S_d$ of $d$ elements. In this way, the permutation group acts in Fock space and plays the same role as the exchange of particles in the $1$st-quantized picture. For ordinary statistics, we have the usual sign change, i.e., $\Delta_d(\sigma)\ket{1,1,\dots,1}=(\pm)^{\sigma}\ket{1,1,\dots,1}$, where $(..)^{\sigma}$ denotes the parity of permutation ($+1$ for bosons and $(-1)^{\sigma}$ for fermions). Nevertheless, the permutation of modes is only a discrete subgroup of the group of single-particle transformation, thus insufficient for the whole physical picture. For example, in our case, it is the subgroup of the unitary group, i.e., $S_d<U(d)$. But it can also be a subgroup of some other group, such as an orthogonal group, in which case one gets parastatistics~\cite{ryan1963, Stoilova_2008}. Therefore, to fully understand how different types of particles integrate into multi-particle states in Fock space, one must study transformation properties under the action of the whole group of single-particle transformations. This work concerns $U(d)$ as our premise is that standard unitary quantum mechanics governs the physics of one particle.  
\subsection{Physical consequences}
\begin{figure}
    \centering
    \includegraphics[scale=0.5]{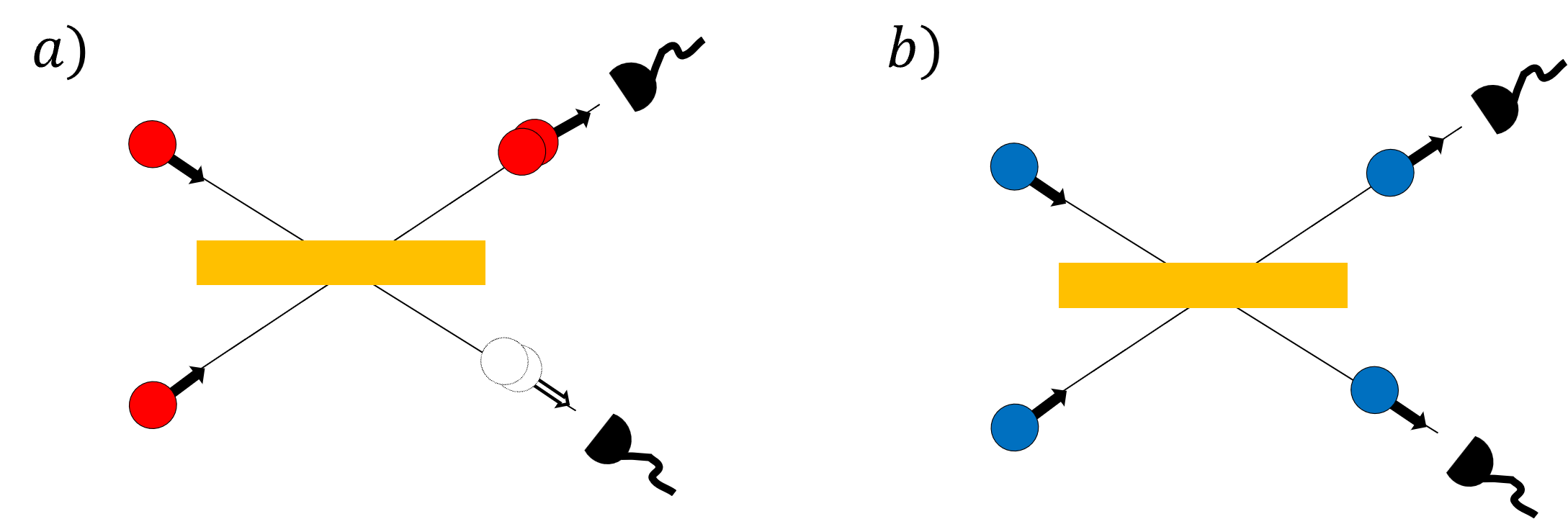}
    \caption{{\bf Hong-Ou-Mandel effect~\cite{HOM}}. a) Boson bunching, and b) fermion antibunching. See main text for details.}
    \label{HOM}
\end{figure}
To illustrate how single-particle transformations affect the physical behavior of indistinguishable particles, take an example of two particles entering the $50/50$ beamsplitter (BS) at different ports (modes), as shown in Fig. \ref{HOM}. The beam-splitter is defined via unitary matrix $u_{bs}=\frac{1}{\sqrt{2}}\begin{pmatrix}
  1 & 1\\ 
  1 & -1
\end{pmatrix}$. 
Now, if particles are bosons, then the input state is $\ket{1,1}=a_1^{\dagger}a_2^{\dagger}\ket{0,0}$, where $a_{1(2)}^{\dagger}$ are bosonic ladder operators associated to two different modes (ports of BS). The output state (after BS) is given by $\frac{1}{2}(a_1^{\dagger}+a_2^{\dagger})(a_1^{\dagger}-a_2^{\dagger})\ket{0,0}=\frac{1}{\sqrt{2}}(\ket{2,0}-\ket{0,2})$. We see that bosons exit 
the BS bunched together, and this is the well-known Hong-Ou-Mandel effect~\cite{HOM}. In contrast, if particles were fermions, the calculation remains the same but with fermionic ladder operators $a_{1/2}^{\dagger}$, thus we have the output state $\frac{1}{2}(a_1^{\dagger}+a_2^{\dagger})(a_1^{\dagger}-a_2^{\dagger})\ket{0,0}=-a_1^{\dagger}a_2^{\dagger}\ket{0,0}=-\ket{1,1}$. This means that fermions exit the BS antibunched (in different ports). These two complementary behaviors can be deduced from the decomposition of the Fock space \eqref{Fock space} into IR sectors of the $U(2)$ group and action of the $u_{bs}$ element. In the case of two bosons, $U(2)$ reduces into three-dimensional subspace $\mathrm{span}\{\ket{2,0},\ket{1,1},\ket{2,0}\}$ (bosonic IR) which encompasses the bunching effect. For two fermions, we have the one-dimensional IR spanned by $\{\ket{1,1}\}$ (fermionic IR), directly resulting in fermionic antibunching.
\subsection{Particle statistics}
As explained at the beginning of this section, the group of single-particle transformations determines the physical behavior of non-interacting indistinguishable particles, and different types of particle statistics arise due to the Fock space's $U(d)$-IR decomposition. Therefore, what we mean by classification of particle statistics is a \emph{classification of all possible ways} the Fock space \eqref{Fock space} decomposes into IR sectors, i.e.
\begin{equation}\label{Fock space IR decomposition}
    \mathcal{F}_d=\bigoplus_{\lambda}c_{\lambda}\mathcal{V}_{\lambda},
\end{equation}
where $\mathcal{V}_{\lambda}$ is an $U(d)$-IR. These are indexed~\cite{fulton} by a partition (Young diagram) $\lambda=(\lambda_1,\dots \lambda_d)$ with $\lambda_1\geq\dots\geq\lambda_d$, and $c_{\lambda}\in\mathbb{N}_0$ is the frequency of the IR. Now, recall that the character of a representation completely determines its decomposition into IR sectors. A well-known fact from representation theory is that IRs of $U(d)$ have \emph{Schur polynomials} $s_\lambda(\vec{x})$ as characters (see Appendix \ref{Schur polynomails} for definition)~\footnote{Note that this holds for polynomial representations of $U(d)$, which is our case here, as the $U(1)$ generator satisfies $\tilde{N}\geq0$. In the most general case, irreducible representations are rational and captured by Laurent-Schur polynomials~\cite{stanley_fomin_1999}. More details are provided in the Section \ref{Negative occupations}.}. Thus, equation \eqref{Fock space IR decomposition} translates to decomposition of character \eqref{character} into Schur-polynomials, i.e.
\begin{equation}\label{character decomposition}
  \prod_{k=1}^d\chi_1(x_k)=
  \sum_\lambda c_\lambda s_\lambda(\vec{x}),~~c_{\lambda}\in\mathbb{N}_0.  
\end{equation}
We see that the single-mode character $\chi_1(x)$ \emph{completely specifies particle statistics} (in the sense of definition \eqref{Fock space IR decomposition}) and this is a direct consequence of our locality assumption \ref{A0}. 

To clarify the point, we provide examples of bosonic and fermionic statistics. For fermions, the maximal occupation number is $p=1$, thus the single-mode character in \eqref{single-mode character} reduces to $\chi_1^{(-)}(x)=1+x$. For $d$-modes, character \eqref{character} can be expanded as 
\begin{eqnarray}\label{fermionic character}
 \chi_d^{(-)}(\vec{x})&=&\prod_{k=1}^d(1+x_k)=1+(x_1+\dots+x_d)+\\\nonumber
&+&(x_1x_2+\dots+x_{d-1}x_d)+\dots+x_1x_1\dots x_d.
\end{eqnarray}
Written in terms of Schur-polynomials, this equation reads
\begin{eqnarray}\label{fermionic Schur exp}\nonumber
 \chi_d^{(-)}(\vec{x})&=&s_{(0,0,\dots,0)}(\vec{x})+s_{(1,0,\dots,0)}(\vec{x})+s_{(1,1,\dots,0)}(\vec{x})+\dots\\
 &+&s_{(1,1\dots,1)}(\vec{x}).
\end{eqnarray}
This expansion corresponds to the decomposition of the Fock space $\mathcal{F}_d=\bigoplus_{N=1}^d\mathcal{V}_{-}^{(N)}$ into fermionic irreducible subspaces $\mathcal{V}_{-}^{(N)}$ associated with particle sectors. 

Similarly, for the case of bosons and $p=+\infty$, equation \eqref{single-mode character} reads $\chi_1^{(+)}(x)=1+x+x^2+\dots=\frac{1}{1-x}$. For $d$ modes, \eqref{character} reads
\begin{eqnarray}\label{bosonic character}
 \chi_d^{(+)}(\vec{x})&=&\prod_{k=1}^d\frac{1}{1-x_k}=1+(x_1+\dots+x_d)\\\nonumber
&+&(x_1^2+x_1x_2+x_2^2+\dots+x_{d-1}x_d+x_d^2)\\\nonumber
&+&(x_1^3+x_1^2x_2+x_1x_2^2+x_2^3+\dots+x_{d-1}x_d^2+x_d^3)\dots
\end{eqnarray}
or written in terms of bosonic Schur polynomials
\begin{eqnarray}\label{bosonic Schur exp}\nonumber
 \chi_d^{(+)}(\vec{x})&=&s_{(0,0,\dots,0)}(\vec{x})+s_{(1,0,\dots,0)}(\vec{x})+s_{(2,0,\dots,0)}(\vec{x})\\
 &+&s_{(3,0,\dots,0)}(\vec{x})+\dots
\end{eqnarray}
Again, this corresponds to the decomposition of the Fock space $\mathcal{F}_d=\bigoplus_{N=1}^{+\infty}\mathcal{V}_{+}^{(N)}$ into bosonic irreducible subspaces $\mathcal{V}_{+}^{(N)}$ associated with particle sectors. 

An important remark is in order about the single-particle sector $\mathcal{F}_d^{(1)}=\mathrm{span}\{\ket{n_1,n_2,...,n_d}~|~\sum_k n_k=1\}$ which is $d$-dimensional. This subspace is associated with the character $s_{(1,0,\dots,0)}(\Vec{x})=x_1+\dots+x_d$, same for bosons and fermions, i.e. we have $\mathcal{F}_d^{(1)}=\mathcal{V}_{+}^{(1)}=\mathcal{V}_{-}^{(1)}$. This is consistent with the fact that the quantum physics of one particle is insensitive to the type of statistics. This also agrees with our operational setup in Fig. \ref{setup}$a)$, which was defined through $d\times d$ unitary matrices acting in the space of one particle. Such representation is called the \emph{standard} or \emph{defining} representation.    
\subsection{Partition theorem and general statistics}
Generally, not all $U(1)$-characters $\chi_{1}(x)=\sum_{s\in\mathbb{N}_0}a_sx^s$ induce a valid $U(d)$-character in \eqref{character}. To see this, take a simple example of $\chi_1(x)=1+x^2$. For two-modes equation \eqref{character} reads $(1+x_1^2)(1+x_2^2)=s_{(0,0)}(x_1,x_2)+s_{(2,0)}(x_1,x_2)+s_{(2,2)}(x_1,x_2)-s_{(1,1)}(x_1,x_2)$ and we have a negative expansion coefficient $c_{(1,1)}<0$, which contradicts $c_{\lambda}\geq0$ in equation \eqref{Fock space IR decomposition}. 

Next, suppose that the first $k$ coefficients in the single-mode character expansion are zero. Then, we can always write $\chi_1(x)=\sum_{s=k}^{+\infty}a_sx^s=x^k \sum_{s=0}^{+\infty}a_{s+k}x^s=x^k\tilde{\chi}_1(x)$. For the general $d$ mode character in \eqref{character}, we will have 
\begin{equation}
    \chi_d(\vec{x})=(x_1\dots x_d)^k\tilde{\chi}_d(\vec{x}).
\end{equation}
The term $(x_1\dots x_d)^k=(\det g)^k$ equals determinant of a unitary matrix $g$ with eigenvalues $x_1,\dots x_d$. From here, we recognize that $\chi_d$ and $\tilde{\chi}_d$ are equivalent up-to-determinant. Therefore, without loss of generality, we will assume $a_0>0$. 

The problem of classifying all single-mode characters that induce valid representation of $U(d)$ involves non-trivial mathematics. Luckily, we found an equivalent formulation to the well-studied combinatorial problem of characterizing completely-positive sequences~\cite{bump, Schoenberg,Davydov, integer}. Details are provided in the Appendix \ref{main proof} together with the proof of our main theorem: 
\begin{theorem}[Partition]\label{Partition theorem}
For $\chi_1(x)=\sum_{s\in\mathbb{N}_0}a_sx^s$ with $a_0>0$, a symmetric function $\prod_{k=1}^d \chi_1(x_k)$ is a $U(d)$ character for all $d\in\mathbb{N}$ if and only if the generating function is of the form
\begin{equation}\label{gen_decompostion}
\chi_1(x)=\frac{Q_{-}(x)}{Q_+(x)},
 \end{equation}
where $Q_{\pm}(x)$ is an integral polynomial with all positive (negative) roots.  Furthermore $Q_+(0)=1$.
\end{theorem}
In other words, $Q_{\pm}(x)=c_{\pm}\prod_{i}(1\mp \alpha_i x)$ are polynomials with integer coefficients, where $\alpha_1>\alpha_2>\dots>0$, $c_{+}=1$ and $c_{-}\in\mathbb{N}$. From here, it follows that $Q_{\pm}(x)$ is a polynomial with all non-zero coefficients. 

Note that we are interested only in \emph{elementary statistics}, i.e., the Fock space cannot be factorized as a tensor product $\mathcal{F}=\mathcal{F}_1\otimes \mathcal{F}_2$, with $\mathcal{F}_{1/2}$ being associated with different particle types. Therefore, the character of elementary statistics cannot be factorized as $\chi_1(x)=\mu_1(x)\nu_1(x)$, with $\mu_1(x)$ and $\nu_1(x)$ being of the type \eqref{gen_decompostion}. Thus equation \eqref{gen_decompostion} for elementary statistics is either $\chi_1=Q_{-}$ or $\chi_1=1/Q_{+}$. We conclude that statistics is of two kinds, i.e., \emph{fermionic-like} $[\dots]_{-}$ and \emph{bosonic-like} $[\dots]_{+}$ specified by
\begin{equation}\label{statistics}
 Q_{\pm}(x)=\sum_{s=0}^{\mathrm{deg}[Q_{\pm}]}(\mp1)^sq_sx^s:=[q_0,q_1,\dots]_{\pm},~~q_s\in\mathbb{N},   
\end{equation}
with $Q_{\pm}(x)$ being irreducible polynomials over integers satisfying conditions in \eqref{gen_decompostion}. The corresponding single-mode characters are $Q_{-}(x)$ and $1/Q_{+}(x)$, respectively. This classification naturally generalizes ordinary statistics, and we term it \emph{transtatistics} with two possible types: \emph{transfermions} (type $[\dots]_{-}$) and \emph{transbosons} (type $[\dots]_{+}$). Here $\mathrm{deg}[Q_{\pm}]$ is the degree of $Q_{\pm}(x)$ to which we also refer as \emph{order of statistics}. Order $0$ is a trivial case, thus we assume $\mathrm{deg}[Q_{\pm}]\geq1$. For $[\dots]_{-}$ statistics, the generalized Pauli principle applies with $p=Q_{-}(1)-1<+\infty$ being the maximal number of particles per mode, while for $[\dots]_{+}$ we have $p=+\infty$. From now on, we shall use the label  $[q_0,q_1,\dots]_{\pm}$ to refer to a particular type of particle statistics.

Note that one can find the eigenvalues of the excitation operator $\Tilde{N}$ defined in \eqref{number operator} by solving the following equation
\begin{equation}\label{f-numbers equation}
    x^{f_0}+x^{f_1}+\dots+x^{f_p}=\left(Q_{\pm}(x)\right)^{\mp1}.
\end{equation}
\section{Irreducible particle sectors: bosons and fermions}\label{bosons and fermions}
Ordinary statistics is order-one statistics of the type $[1,1]_{\pm}$. To answer what makes bosons and fermions special in the whole family of generalized statistics classified in \eqref{statistics}, we introduce the following assumptions:
\begin{Assumption}[{\bf Irreducibility}]\label{A1} \emph{All symmetries of the system of indistinguishable particles are determined by the $U(d)$ group.}
\end{Assumption}
Assumption \ref{A1} essentially states that the Fock space decomposes into $U(d)$-IR sectors without multiplicity; thus, no additional symmetries (conserved quantities) are present in the system. We show now that only ordinary statistics has this property. 

We start with a general single-mode character $\chi_1(x)=\sum_{s=0}^{+\infty}a_sx^s$. Character equation \eqref{character} for $d$-modes can be expanded as follows
\begin{eqnarray}\label{Schur-decomposition}
    \chi_d(\Vec{x})&=&a_0^d+a_0^{d-1}a_1(x_1+\dots+x_d)+W(\Vec{x})\\\nonumber
    &=&a_0^ds_{(0,0,\dots,0)}(\vec{x})+a_0^{d-1}a_1s_{(1,0,\dots,0)}(\Vec{x})+W(\Vec{x}),
\end{eqnarray}
where $W(\Vec{x})$ is the symmetric function that contains quadratic and higher-order terms in variables $\vec{x}=(x_1,\dots,x_d)^T$. Since Schur polynomials of degree $l$ form the basis in the space of $l$-degree symmetric polynomials, the constant and linear terms in the equation \eqref{Schur-decomposition} are already IR-decomposed. Because assumption \ref{A1} requests no multiplicities, we have $a_0=0,1$ and $a_1=0,1$.   

Now we turn to concrete cases. For transfermions, the single-mode character reads
\begin{eqnarray}\label{transF character}
    Q_{-}(x)&=&c_{-}\prod_i(1+\alpha_ix)\\\nonumber
    &=&c_{-}+c_{-}\left(\sum_i\alpha_i\right) x+c_{-}\left(\sum_{i< j}\alpha_i\alpha_j\right)x^2+\dots\\\nonumber
    &=&a_0+a_1x+a_2x^2+\dots 
\end{eqnarray}
    with $\alpha_1>\alpha_2>\dots>0$ and $c_{-}\in\mathbb{N}$.  
This is consistent with the previous analysis of \eqref{Schur-decomposition} only if $c_{-}=a_0=1$ and $\sum_i\alpha_i=a_1=1$. For the quadratic coefficient in \eqref{transF character} we have $a_2=\sum_{i< j}\alpha_i\alpha_j=\frac{1}{2}(\sum_i\alpha_i)^2-\frac{1}{2}\sum_i\alpha_i^2=\frac{1}{2}-\frac{1}{2}\sum_i\alpha_i^2\in\mathbb{N}_0$ because $Q_{-}$ is an integral polynomial. This is possible only if  $\alpha_1=1$ and $\alpha_2=\dots=0$. Thus, we recover the fermionic character $\chi_1(x)=1+x$. 

For the case of transbosons, we have the single-mode character 
\begin{eqnarray}\label{transB character}
    1/Q_{+}(x)&=&1/\prod_i(1-\alpha_ix)\\\nonumber
    &=&1+\left(\sum_i\alpha_i\right) x+\left(\sum_i\alpha_i^2+\sum_{i< j}\alpha_i\alpha_j\right)x^2+\dots\\\nonumber
    &=&a_0+a_1x+a_2x^2+\dots 
\end{eqnarray}
By the same analysis as for transfermions, we conclude $\sum_i\alpha_i=1$. For the quadratic term in \eqref{transB character}, we have $a_2=\sum_i\alpha_i^2+\sum_{i< j}\alpha_i\alpha_j=\frac{1}{2}(\sum_i\alpha_i)^2+\frac{1}{2}\sum_i\alpha_i^2=\frac{1}{2}+\frac{1}{2}\sum_i\alpha_i^2\in\mathbb{N}_0$. Again, this is satisfied only if $\alpha_1=1$ and $\alpha_2=\dots=0$. Thus $\chi_1(x)=\frac{1}{1-x}$ and we recover the bosonic character. This concludes that only bosonic and fermionic statistics are consistent with the assumption \ref{A1}.

For ordinary statistics, the excitation operator in \eqref{number operator} coincides with the standard number operator. The Casimir operator in \eqref{Casimir} becomes the total number of particles which is a conserved quantity linked with $N$-particle sectors
\begin{equation}\label{particle sectors}
 \mathcal{F}_d^{(N)}=\mathrm{span}\{\ket{n_1,n_2,...,n_d}~|~\sum_k n_k=N\}.  
\end{equation} 
These are also $U(d)$-IR sectors associated with the standard bosonic (fermionic) subspaces $\mathcal{V}_{\pm}^{(N)}$.

It is worth pointing out that only in the case of ordinary statistics is the solution to the equation \eqref{f-numbers equation} for spectrum $f_n$ of the excitation operator $\Tilde{N}$ non-degenerate (in this case $f_n=n$). In all other cases, degeneracy necessarily appears. This follows from the fact that coefficients in the polynomial $Q_{\pm}(x)$ are all non-zero, and at least one of them is $2$ or greater (otherwise, all coefficients are equal to $1$, and we have ordinary statistics). Given this, at least one expansion coefficient on the right-hand side of \eqref{f-numbers equation} is $2$ or greater. Thus at least two $f_n$ numbers on the left-hand side of \eqref{f-numbers equation} are the same. However, it should be noted that these results are based on the assumption that $\Tilde{N}\geq0$. A generalization to the possibility of negative values of $\Tilde{N}$ is discussed in Section \ref{Negative occupations}.
\section{Hidden symmetry and transtatistics}\label{order-one statistics}
We learned from the previous analysis that multiplicities in the Fock space decomposition \eqref{Fock space IR decomposition} will necessarily appear for all transtatistics apart from bosonic and fermionic. These multiplicities cannot be resolved without additional, so-called \emph{hidden symmetry}, present in the system~\cite{mcintosh1959accidental}. The latter is typically identified as a higher symmetry of the Hamiltonian required to fully resolve the degeneracy of the energy spectrum (sometimes called `accidental' degeneracy). The classic example is the degeneration of the spectrum of the hydrogen atom not captured by the rotational symmetry ($\mathrm{SO}(3)$ group) of the Hamiltonian but requires a higher (hidden) symmetry for resolution, which is the $\mathrm{SO}(4)$ group~\cite{pauli1926wasserstoffspektrum}. In our case, the situation is similar; the multiplicities in the Fock space decomposition are in one-to-one correspondence with the degeneration of the Hamiltonian $\hat{H}=\sum_{^k=1}^d\epsilon_k\Tilde{N}_k$ defined in \eqref{Hamiltonian} (generator of the $U(d)$ action). The total energy is given by
 \begin{equation}\label{energy spectrum}
     E=\sum_{k=1}^d\epsilon_k f_k,
 \end{equation}
 with $f_k$ being the eigenvalues of the excitation operator $\Tilde{N}$ defined in \eqref{number operator}. As long as this operator is non-degenerate, the energy spectrum $E$ is well-resolved with the set of quantum numbers $(f_{k_1},\dots,f_{k_d})$. Nevertheless, we have seen that this happens only in the case of ordinary statistics. For all other cases, degeneracy in spectrum $f_n$ necessarily appears, which is to be resolved by different quantum numbers unrelated to the $U(d)$ group. Without the specification of these numbers, the representation of $U(d)$ in Fock space remains unspecified, defined only up to IR-multiplicity. 

We will study these effects in detail for the first non-trivial case beyond ordinary statistics, i.e., the order-one statistics $[1,q]_{\pm}$, with $q\in\mathbb{N}$. To make the analysis more accessible, we will separate the notation for transbosons 
 (type $[1,\beta]_{+}$) and transfermions (type $[1,\alpha]_{-}$) with $\alpha,\beta\in\mathbb{N}$. The reason why we set the first coefficient in $[\dots]_{\pm}$ to be $1$ is because we will restrict our analysis only to the case of a unique vacuum state $\ket{0}^{\otimes d}$. To be more precise, we will study the cases in which the only invariant state under $U(d)$ is a vacuum state. This is possible only if the first coefficient in the single-mode character $\chi_1(x)=\sum_{s=0}^{+\infty}a_sx^s$ is set $a_0=1$ (see discussion around equation \eqref{Schur-decomposition}).   

To begin with, take an example of transfermions $[1,\alpha]_{-}$ with $\alpha=2$. In this case, the singe-mode Fock space is three-dimensional (follows from $\chi_1(x)=1+2x=1+2e^{i\theta}$), and the maximal occupation number is $p=\alpha=2$. For the case of two modes, the character reads $\chi_2(x_1,x_2)=1+2(x_1+x_2)+4x_1x_2=s_{(0,0)}(x_1,x_2)+2s_{(1,0)}(x_1,x_2)+2^2s_{(1,1)}(x_1,x_2)$. Thus, the Fock space decomposes into fermionic multiplets of the size $\alpha^N=2^N$, for $N=0,1,2$. This exponential growth of multiplicity is generic to order-one transtatistics. It is formalized in the following theorem (see Appendix \ref{proof of Fock space dec} for proof)

\begin{theorem}\label{Polarity}
Fock-spaces for $[1,\alpha]_{+}$ and $[1,\beta]_{-}$ decompose into IR sectors as 
\begin{align}\label{transF}
  \mathcal{F}_d^{}&=\bigoplus_{N=0}^d \alpha^N \mathcal{V}_{-}^{(N)},~~\alpha\in\mathbb{N},\\\label{transB}
  \mathcal{F}_d^{}&=\bigoplus_{N=0}^{+\infty} \beta^N \mathcal{V}_{+}^{(N)},~~\beta\in\mathbb{N},
\end{align}

where $\mathcal{V}_{-}^{(N)}$ and $\mathcal{V}_{+}^{(N)}$ are the fermionic and bosonic IRs ($N$-particle sectors for ordinary statistics), respectively. 
\end{theorem}
 In the next section, we will build the concrete ansatz to identify auxiliary quantum numbers to resolve the degeneracy in \eqref{transF}-\eqref{transB}. Based on this, we will construct the $U(d)$ representation in Fock space. 
\subsection{Hidden quantum numbers}\label{hidden symmetry}

We start with transfermions $[1,\alpha]_{-}$ for some $\alpha\geq2$. For this case, the single-mode character reads $\chi_1(x)=1+\alpha x$ with $x=e^{i\theta}\in U(1)$. The single-mode Fock space is $(\alpha+1)$-dimensional $\mathcal{F}_1=\mathrm{span}\{\ket{n}~|~n=0,1,2,\dots ,p=\alpha\}$. Equation \eqref{f-numbers equation} reads 
\begin{equation}\label{fermi character equation}
    x^{f_0}+x^{f_1}+\dots+x^{f_p}=1+\alpha x,
\end{equation}
with the solution $f_0=0$ and $f_n=1$ for $n=1,\dots, \alpha$. The generator of $U(1)$ action $\Delta_1(e^{i\theta})=e^{i\theta\tilde{N}}$ is the excitation operator defined in \eqref{number operator} and in our case, it acts is as follows
\begin{equation}\label{N tansF}
    \Tilde{N}\ket{n}=\begin{cases} 
      0 & n = 0, \\
      +1\ket{n} & n=1,\dots, \alpha.
   \end{cases}
\end{equation}

Given this, one can re-interpret the single-mode states $\ket{n}$ for $n\geq1$ as \emph{de facto} being the single-particle excitations distinguished by some auxiliary degree of freedom with $\alpha$ values. Therefore, we can introduce decomposition $n=k+z$ with $k=0,1$ being the `real' occupation number of the fermionic type and $z=0,\dots \alpha^k-1$ as an auxiliary quantum number accounting for degeneracy. Having this, the formula \eqref{N tansF} takes the standard form, i.e $\Tilde{N}\ket{k+z}=k\ket{k+z}$. Now, to separate degrees of freedom captured by $k$ and $z$ quantum numbers, we introduce the mapping
\begin{equation}
    L_1\ket{k+z}=\begin{cases} 
      \ket{0}_F & k = 0, \\
      \ket{1}_F\otimes \ket{z}_A & k=1,
   \end{cases}
\end{equation}
where $\ket{k}_F$ is the ordinary fermionic number state with $k=0,1$, while $\ket{z}_A$ (with $z=0,\dots \alpha-1$) is a new degree of freedom emerged solely from the \emph{statistics type}. The ansatz straightforwardly generalizes to the $d$-mode Fock space. We define 
\begin{equation}\label{fermion mapping}
    L_d\ket{n_1,\dots,n_d}=\mathcal{T}L_1^{\otimes d}\ket{n_1,\dots,n_d},
\end{equation}
where $\mathcal{T}$ is the shift operator needed to separate degrees of freedom, i.e., to shift all auxiliary states to the right. For example, $\mathcal{T}\ket{k_1}\ket{z_1}\ket{k_2}\ket{z_2}=\ket{k_1,k_2}_F\otimes\ket{z_1,z_2}_A$. To fully clarify the mapping in \eqref{fermion mapping}, let $\ket{n_1,\dots, n_d}=\ket{k_1+z_1,\dots, k_d+z_d}$, where again, $k_s=0,1$ and $z_s=0,\dots,\alpha^{k_s}-1$. We form the ordered list $(z_{s_1},\dots,z_{s_N})$ for which $k_{s_r}=1$, i.e. the list of all non-zero fermionic excitations. Here $N=d-(\delta_{0,n_1}+\dots+\delta_{0,n_d})$ is the total number of them. Then, the equation \eqref{fermion mapping} reads 
\begin{equation}
    L_d\ket{n_1,\dots,n_d}=\ket{k_1,\dots, k_d}_F\otimes \ket{z_{s_1},\dots,z_{s_N}}_A
\end{equation}
where $\ket{k_1,\dots, k_d}_F$ is the ordinary $N$-particle fermionic state, with the auxiliary label of particles $\ket{z_{s_1},\dots,z_{s_N}}_A$. This brings us precisely to the decomposition in \eqref{transF}, which can also be written as 
\begin{equation}\label{fermi space}
    \mathcal{F}_d=\bigoplus_{N=0}^d \mathcal{V}_{-}^{(N)}\otimes \mathcal{H}_A^{\otimes N},
\end{equation}
where $\mathcal{H}_A=\mathrm{span}\{\ket{z}|~z=0,\dots \alpha-1\}$ is the auxiliary space. Given this factorization, it is clear that $U(d)$ acts only in the fermionic part $\mathcal{V}_{-}^{(N)}$, while $\mathcal{H}_A^{\otimes N}$ remains untouched. The additional $U(\alpha)$ group acting in the space $\mathcal{H}_A^{(N)}$ can be added to resolve degeneracy completely. Now, for an element $g\in U(d)$, let the standard action on the fermionic number state is $\Delta_d^{(F)}(g)\ket{k_1,\dots, k_d}_F$. This induces the action $\Delta_d(g)$ in the Fock space \eqref{Fock space} as
\begin{equation}\label{fermi rep}
 \Delta_d(g)=L_d^{-1}\left(\Delta_d^{(F)}(g)\otimes \openone_A \right)L_d,   
\end{equation}
where $L_d$ is the mapping given in\eqref{fermion mapping}. With this, we have defined the action of $U(d)$ in the Fock space.

In complete analogy, we provide an ansatz for transbosons of $[1,\beta]_{+}$ type with $\beta\geq2$. In this case, we have the single mode character $\chi_1(x)=\frac{1}{1-\beta x}$ with $x=e^{i\theta}\in U(1)$ and the single-mode Fock space is infinte-dimensional $\mathcal{F}_1=\mathrm{span}\{\ket{n}~|~n=0,1,2,\dots\}$. As before, we shall solve equation \eqref{f-numbers equation} 
\begin{equation}\label{bose-character equation}
    x^{f_0}+x^{f_1}+x^{f_2}+\dots=\frac{1}{1-\beta x}.
\end{equation}
It is convenient to write the particle number $n$ in the form 
\begin{equation}
    n=1+\beta^2+\dots+\beta^{k-1}+z=\frac{\beta^k-1}{\beta-1}+z,
\end{equation}
with $k=0,1,2,\dots$ and $z=0,\dots, \beta^{k}-1$. Having this notation, the solution to \eqref{bose-character equation} is  simple, i.e., $f_{\frac{\beta^k-1}{\beta-1}+z}=k$. We have the following action of the excitation operator $\Tilde{N}$
\begin{equation}\label{bose U(1)}
    \Tilde{N}\ket{\frac{\beta^k-1}{\beta-1}+z}=k\ket{\frac{\beta^k-1}{\beta-1}+z}.
\end{equation}
Here $k$ represents the `new' occupation number of the bosonic type, while $z$ is an auxiliary quantum number. Since $z=0,\dots, \beta^{k}-1$ counts all possible states associated to $k$ bosonic excitations, it is  convenient to write $z$ in the $\beta$-base, i.e. $z=z_{k-1}\beta^{k-1}+z_{k-2}\beta^{k-2}+\dots+z_{0}\beta^0:=z_{k-1}z_{k-2}\dots z_{0}$, where $z_s=0,\dots,\beta-1$ are the digits. With this, we can introduce the mapping
\begin{equation}
    L_1\ket{\frac{\beta^k-1}{\beta-1}+z}=\begin{cases} 
      \ket{0}_B & k = 0, \\
      \ket{k}_B\otimes \ket{z_{k-1}z_{k-2}\dots z_{0}}_A & k>0,
   \end{cases}
\end{equation}
where $\ket{k}_B$ is the ordinary bosonic Fock (number) state with $k=0,1,2,\dots$, while $\ket{z_{k-1}z_{k-2}\dots z_{0}}_A$ (with $z_s=0,\dots, \beta-1$) is associated to the statistics degree of freedom. The generalization to the $d$-mode Fock state is as for the case of transfermions, i.e., we use the same equation \eqref{fermion mapping}. In this case, we have
\begin{equation}\label{bose mapping}
    L_d\ket{n_1,\dots,n_d}=\ket{k_1,\dots, k_d}_B\otimes \ket{z_{s_1},\dots,z_{s_N}}_A,
\end{equation}
where $\ket{k_1,\dots, k_d}_B$ is the ordinary bosonic number state with $k_s=0,1,2,\dots$, while $\ket{z_{s_1}\dots z_{s_N}}_A$ comes from type of statistics. The action of $U(d)$ is introduced in complete analogy to the fermionic case and equation \eqref{fermi rep}.
\subsection{Is hidden symmetry an ordinary internal symmetry?}
We may question if the hidden quantum numbers introduced in the previous section are related to some genuine degree of freedom emerging from the type of statistics. Could these numbers be associated with the standard internal degrees of freedom, such as spin? For example, degeneration in \eqref{energy spectrum} could be potentially explained by the argument that energy is spin-independent, and then, transatistics may be just an ordinary (fermionic of bosonic) statistics where $U(d)$ affects only external degrees of freedom (such as modes represented by paths of particles in Fig. \ref{setup}$a$). However, this argument cannot be well-aligned with the Fock-state decomposition in \eqref{transF}-\eqref{transB}, even though only multiplets of ordinary statistics appear in decomposition.
This is due to the dimension discrepancy between ordinary statistics and transtatistics. To see this, suppose that we deal with ordinary fermions with $d$ real degrees of freedom ($d$ modes on which $U(d)$ acts) and some internal degree of freedom (e.g., spin) with $z=0,\dots, \alpha-1$ values, which is unaffected by $U(d)$. The overall dimension of the single-particle space is $\alpha d$; hence the dimension of the fermionic Fock space is $2^{\alpha d}$. This starkly contrasts the dimension $\alpha^d$ of the transfermionic Fock space for $[1,\alpha]_{-}$ type. As we shall see, this dimension discrepancy will differentiate the thermodynamics of non-interacting systems of ordinary and transtatistics. The latter will be accompanied by the effect of generic spontaneous symmetry breaking absent in ordinary statistics. 

Note that analogy to, e.g., spin degree of freedom discussed here is only possible for order-one statistics. For higher-order statistics, no (obvious) similarities can be concluded. We will discuss this point later.
\subsection{Relation to thermodynamics}
To study thermodynamics, we consider the single-particle energy spectrum $\epsilon_1,\dots, \epsilon_d$, where $\epsilon_k$s represent the energies associated with different modes. This situation is similar to the one discussed in Section \ref{Generalized N number}. In that section, we examined the unitary evolution generated by the Hamiltonian given in equation \eqref{Hamiltonian}. However, the system is in contact with a thermal bath in the present case. The thermodynamical quantities (e.g., for canonical ensemble) can be derived from the partition function $Z_{d}(\beta)=\mathrm{Tr}e^{-\beta \hat{H}}$, and its explicit form follows directly from the form of $H$, i.e., 
\begin{equation}\label{partition function}
Z_d(\beta)=\prod_{k=1}^dZ_1(e^{-\beta\epsilon_k})=\prod_{k=1}^d\chi_1(e^{-\beta\epsilon_k})
\end{equation}
with $\beta=1/k_BT$ being the Boltzmann factor. 

The physical relevance of character $\chi$ can also be understood through thermodynamics~\cite{Balantekin_2001}. This is because we can get the partition function from the character via Wick's rotation, that is, $it/\hbar\epsilon_k\rightarrow -\beta\epsilon_k$. In this respect, the product form of equation \eqref{partition function} arises directly from our central assumption \ref{A0}, i.e., the overall partition function can be expressed as a product of individual partition functions (associated with individual modes). This aligns with the expected behavior for independent systems, such as a set of independent modes. Therefore, assumption \ref{A0} is in one-to-one correspondence with the independence in a thermodynamical sense. The formula \eqref{partition function} trivially holds for ordinary statistics (bosons and fermions)~\cite{ashcroft}.

When the system is capable of exchanging excitations (particles) with a reservoir, we can analyze its behavior using the grand canonical partition function $\mathcal{Z}_d=\mathrm{Tr}e^{-\beta(\hat{H}-\mu\tilde{N})}$. In this expression, $\tilde{N}$ represents the excitation operator as defined in equation \eqref{number operator}, and $\mu$ corresponds to the chemical potential (variable conjugated to $\Tilde{N}$).
The explicit form of the grand canonical partition function $\mathcal{Z}_d$ is as follows
\begin{equation}\label{grand partition func}
\mathcal{Z}_d=\prod_{k=1}^dZ_1(e^{-\beta(\epsilon_k-\mu)})=\prod_{k=1}^d\chi_1(e^{-\beta(\epsilon_k-\mu)}).    
\end{equation}

\subsection{Thermodynamics of ideal gasses and spontaneous symmetry breaking}\label{TD}
\begin{figure}
    \centering
    \includegraphics[scale=0.6]{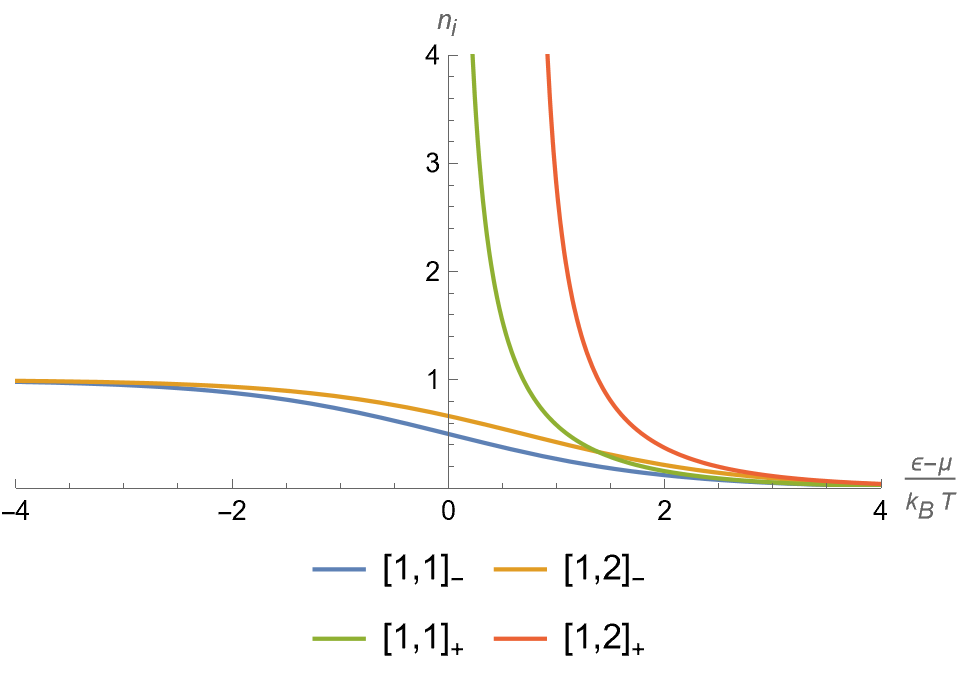}
    \caption{{\bf Mean particle number} for ordinary (blue and green) and generalized (orange and red) statistics.}
    \label{mean particle number}
\end{figure}
Let us examine the thermodynamical properties of a non-interacting system for general order-one statistics $[1,q]_{\pm}$ with $q\in\mathbb{N}$. Ordinary statistics is recovered for $q=1$. We consider a grand-canonical ensemble defined by a set of single-particle energies $\epsilon_1,\dots, \epsilon_d$ associated with different modes. The system is described by a equilibrium state $\rho=\frac{1}{\mathcal{Z}_d}e^{-\beta(\hat{H}-\mu\Tilde{N})}$, where $\mathcal{Z}_d$ is the grand-canonical partition function defined in \eqref{grand partition func}. All thermodynamical quantities can be evaluated from the grand canonical potential $\Omega=-\frac{1}{\beta}\log \mathcal{Z}_d$. For example, $N=\frac{\partial\Omega}{\partial\mu}$ gives the mean particle number. For the case of transtatistics $[1,q]_{\pm}$, we get 
\begin{equation}\label{mean number}
    N=\sum_k n_k=\sum_i\frac{1}{\frac{1}{q}e^{\beta(\epsilon_k-\mu)}\mp1}.
\end{equation}
This expression reduces to the Fermi-Dirac and Bose-Einstein distributions for $q=1$. The plots for $n_i$ in \eqref{mean number} for various statistics are presented in Fig. \ref{mean particle number}. For the fermionic-type statistics, equation \eqref{mean number} reduces to the Fermi-Dirac distribution $n_k=\theta(\mu-\epsilon_k)$ at zero temperature for all $q$. For the bosonic type, the mean number diverges at the values of energy $\epsilon=\mu+\frac{1}{\beta}\log q$ when the Bose-Einstein condensation occurs. In the classical limit of $\beta(\epsilon-\mu)\gg1$, the formula \eqref{mean number} reduces to the standard Maxwell-Boltzmann distribution, i.e., $n_k\approx q e^{-\beta(\epsilon_k-\mu)}$, where the factor $q$ appears as the degeneracy factor. The same factor appears in the classical limit for standard quantum gasses with $q=2s+1$ coming from spin $s$ (see, for example, Chapter 8.3. in \cite{salinas2013}). This is because the energy is independent of spin, and thus, the energy spectrum degenerates. 
    
Note that the chemical potential $\mu$ in the formulas above is temperature dependent. To be more precise, the standard approach to thermodynamics of ideal gasses is to keep total particle number $N$ as a fixed parameter and then invert \eqref{mean number} to calculate chemical potential $\mu=\mu(N, T)$ as a function of a total number of particles and temperature~\cite{ashcroft}. Given this, one can introduce a simple change of variables $\mu\rightarrow\mu-\frac{1}{\beta}\log q$, and the formula \eqref{mean number} would reduce to one for ordinary statistics. This means that solution for the chemical potential for order-one transtatistics is 
\begin{equation}\label{chmical potential}
\mu_q=\mu_{q=1}-k_BT\log q,    
\end{equation}
where $\mu_{q=1}$ is the chemical potential of ordinary statistics. What follows is that almost all thermodynamical quantities (e.g., mean energy, heat capacity, etc.) remain the same as in the case of ordinary statistics for arbitrary $q$. Nevertheless, the entropy will change. To see this, note that $S=-\beta \Omega+\beta\langle E\rangle -\beta\mu N$, thus the shift of $-k_BT\log q$ in the chemical potential introduces a change in the entropy, i.e.
\begin{equation}
    S_{q}=S_{q=1}+k_BN\log q.
\end{equation}
The entropy of ordinary statistics $S_{q=1}$ vanishes at $T=0$; hence, a residual entropy of $k_B N\log q $ remains at zero temperature for all $q>1$. This is consistent with the fact that fermionic (bosonic) $N$-particles IRs in the Fock space decomposition \eqref{transF}-\eqref{transB} appear $q^N$ times; therefore, the ground state is $q^N$ times degenerate. This degeneration is known to result in \emph{residual entropy at zero temperature} and is associated with \emph{spontaneous symmetry breaking}~\cite{peierls1991spontaneously}, here present for transtatistics. This is one of the main differences compared to ordinary quantum gasses exhibiting non-degenerate ground states.  
\section{Discussion and outlook}
\subsection{Statistics of higher order} 
Here we briefly analyze some of the technical and conceptual difficulties that arise when dealing with statistics of higher order. As an illustration, we take the example of statistics of order two $[1,q_1,q_2]_{\pm}$. A simple inspection shows that polynomial $Q_{\pm}(x)=1\mp q_1x+q_2x^2$ has non-negative (positive) roots for $q_1^2>4q_2$. To see how the Fock space decomposes in some simple cases, consider transfermions $[1,q,1]_{-}$ and the corresponding two-mode character 
\begin{eqnarray}\nonumber
\chi_2(x_1,x_2)&=&(1+qx_1+x_1^2)(1+qx_2+x_2^2)\\\nonumber
&=&s_{(0,0)}(x_1,x_2)+qs_{(1,0)}(x_1,x_2)+\\\nonumber
&=&(q^2-1)s_{(1,1)}(x_1,x_2)+s_{(2,0)}(x_1,x_2)+\\
&+&qs_{(2,1)}(x_1,x_2)+s_{(2,2)}(x_1,x_2).    
\end{eqnarray}
The IR characters $s_{(2,1)}$ and $s_{(2,2)}$ that are not fermionic nor bosonic type show-up in the decomposition. This is a typical feature that appears for any higher-order statistics. In turn, finding Fock space's decomposition for general $d$ modes, such as one provided for order-one statistics in \eqref{transF}-\eqref{transB}, is more difficult. Next, the dimension of the single-mode Fock space is $q+2$, and the maximal occupation number is $p=q+1$. The solution to the single-mode character equation \eqref{f-numbers equation}
\begin{equation}
    x^{f_0}+\dots +x^{f_{q+1}}=1+qx+x^2
\end{equation}
is $f_0=0$ and $f_{q+1}=2$, while $f_n=1$ for $n=1,\dots,q$. Recall that these are the eigenvalues of the excitation operator $\Tilde{N}$ in \eqref{number operator}, and as we see, we have three distinct values $k=0,1,2$. Again, we have degeneration of the spectrum, but resolving it is a more delicate issue than for the case of order-one statistics we have presented in Section \ref{order-one statistics}.  This is partially because a clear interpretation is missing. For example, we may try to label the single-mode Fock states with two quantum numbers, $k=0,1,2$ (for excitations), and one auxiliary number $z_k$, to account for degeneracy. As before, we have $\tilde{N}\ket{k+z_k}=k\ket{k+z_k}$, with 
$z_k=0$ for $k=0,2$, while for $k=1$ we have $z_k=0,\dots,q$. This appears paradoxical because degeneracy is present for one excitation but disappears for two. From this example, we see the analysis of degeneracy and categorization of hidden quantum numbers becomes significantly more complicated due to the involvement of `non-standard' IRs. 

Let's take a look at the general case of $[1, q_1, \dots, q_m]_{\pm}$, where $Q_\pm(x) = \prod_{s=1}^m (1 \pm \alpha_k x)$ with $\alpha_s > 0$ and $m=\mathrm{deg}[Q_\pm]$ representing the degree of statistics. Given the Cauchy identities \cite{stanley_fomin_1999}  
\begin{widetext}
\begin{eqnarray}\label{FermiCauchy}
    \prod_{k=1}^d\prod_{s=1}^m (1+\alpha_sx_k)&=&\Sigma_{l(\lambda) \leq m} s_{\lambda}(x_1,...,x_d)s_{\lambda}(\alpha_1,...,\alpha_m),\\\label{BoseCauchy}
    \prod_{k=1}^d\prod_{s=1}^m\frac{1}{1-\alpha_sx_k}&=&\Sigma_{l(\lambda) \leq m} s_\lambda(x_1,...,x_d)s_\lambda(\alpha_1,...,\alpha_m),
\end{eqnarray}
\end{widetext}
where $l(\lambda)$ is the length of the diagram (number of rows), and $l(\lambda')$ is the conjugate partition of $\lambda$, we that the multiplicity in the Fock space decomposition \eqref{Fock space IR decomposition} are given by Schur polynomials evaluated at
the parameters $\alpha_s$, i.e. $a(\lambda)=s_\lambda(\alpha_1,...,\alpha_m)$. This multiplicity coefficient can be simplified using the so-called Kostka numbers~\cite{stanley_fomin_1999}, but we leave this for further consideration due to the involvement of non-trivial combinatorics.      

On the other hand, thermodynamic considerations are simpler due to the product nature of the single-mode character. In the general case $[1,q_1,\dots,q_m]_{\pm}$ the grand canonical partition function \eqref{grand partition func} becomes
\begin{equation}\label{grand partition func1}
\mathcal{Z}_d=\prod_{k=1}^d\prod_{s=1}^m\left(1\mp e^{-\beta(\epsilon_k-\mu)}\right)^{\mp1}.    
\end{equation}
In complete analogy to the derivation of \eqref{mean number}, we get the mean particle number as
\begin{equation}\label{mean number general}
    n_k=\sum_{s=1}^m\frac{1}{\frac{1}{\alpha_s}e^{\beta(\epsilon_k-\mu)}\mp1},
\end{equation}
where $m$ is the degree of statistics. This expression reduces to the Fermi-Dirac and Bose-Einstein distributions for $m=1$ and $\alpha_1=1$. For transfermions, equation \eqref{mean number general} reduces to $n_k=m~\theta(\mu-\epsilon_k)$ at zero temperature. As expected, $m$ (degree of statistics) trasnfermions can occupy the same energy level at $T=0$ as a consequence of the generalized Pauli exclusion principle. For $N$ particles, transfermions will form the analogon of the Fermi sea at zero temperature. On the other hand, we have for transbosons, the mean number diverging at the values of energy $\epsilon=\mu+\frac{1}{\beta}\log \alpha_{\mathrm{max}}$ when the Bose-Einstein condensation occurs. In the classical limit of $\beta(\epsilon-\mu)\gg1$, the formula \eqref{mean number general} reduces to the standard Maxwell-Boltzmann distribution, i.e., $n_k\approx q e^{-\beta(\epsilon_k-\mu)}$, where the factor $q=\sum_s\alpha_s (\beta_s)$ appears as the degeneracy factor.
\subsection{'Negative occupation' numbers, exceptional statistics and extension to an infinite number of modes}\label{Negative occupations}
So far, we have discussed only the case of non-negative eigenvalues of $U(1)$ generators, i.e., the generalized excitation operator as defined in Section \eqref{Generalized N number} and eq. \eqref{Casimir} satisfying $\Tilde{N}\geq0$. Here, we will discuss the extensions to 'negative occupations', allowing $\Tilde{N}$ to have negative eigenvalues, which encounters the most general $U(1)$ generator. On the physical side, such a situation is particularly relevant for the theory of antiparticles introduced for ordinary statistics by Dirac~\cite{dirac1930theory}. 
To do so, we shall extend the single-mode character in \eqref{single-mode character} to a (formal) Laurent series $\chi_{1}(x)=\sum_{s\in\mathbb{Z}}^{}a_sx^s$, with $a_s$ integer coefficients, such that the overall product $\chi_d(\vec{x})=\prod_{k=1}^d\chi_{1}(x_k)$ is a valid $U(d)$-character. In such case, $\chi_d(\Vec{x})$ decomposes over Laurent-Schur polynomials of the form $(x_1\dots x_d)^{-k}s_{\lambda}(x_1,\dots,x_d)$, with $k\in\mathbb{N}_0$, encompassing the most general irreducible representations of $U(d)$~\cite{stanley_fomin_1999}. Details are provided in the Appendix \ref{Negative occupations}, where we prove the following classification theorem: 
\begin{theorem}
    General $U(1)$-character (under conditions previously discussed) is one of the following forms
    \begin{itemize}
        \item [1)] {\bf transtatistical type}, i.e. either $x^k\chi_1(x)$ or $x^k\chi_1(1/x)$, where $\chi_1(x)$ is the character classified by the Theorem \ref{Partition theorem} and $k\in\mathbb{Z}$, or
        \item [2)] {\bf exceptional type} $a\sum_{s\in\mathbb{Z}}\rho^s x^s$, with $a,\rho\in\mathbb{N}_0$.
    \end{itemize}
\end{theorem}
This concludes the classification in the most general case. For the trastatistical types, the factor $x^k$ builds to $(x_1\dots x_d)^k=(\mathrm{det}(g))^k$ for the case of $d$ modes, which corresponds to the determinant representation. Thus, representations can be constructed from representations obtained for transtatistics wired with the determinant representation (power $k$). Note that $\chi_1(1/x)=\chi_1(x)^\ast$, because $x=e^{i\varphi}$, thus this character is associated with the conjugate representation. Finally, the exceptional character is of a fundamentally different form and essentially comes from exceptional totally-positive sequences (see Appendix \ref{Negative occupations} for details). This is a very interesting case in which generalized number operator $\tilde{N}$ is unbounded both from below and above, which is not the case for transtatistical types of statistics.

Finally, we will conclude this section with a brief comment on the generalization to the infinite set of modes when $d=\infty$. This is particularly relevant for investigating the algebra of creation and annihilation operators for transtatiscics field theory. In this respect, establishing a relation to the existing results on $U(\infty)$ representations and extreme characters theory~\cite{vershik1982,okounkov1997,voiculescu1976} is very promising. In particular, striking similarities are found in the generating functions of the so-called ``extreme characters'' and our principal decomposition given in Theorem \ref{Partition theorem}.
\subsection{Relation to other generalized statistics}
An obvious question is if and how our statistics classified in \eqref{statistics} differs from other generalized statistics presented in the literature. Of course, we are not able to exhaustively compare but rather analyze the most common cases. The first remark is that the main difference is due to the underlying symmetries. Our classification relies on the $U(d)$ group, while in most of the cases, other generalized statistics is based on a different group. Take an example of fractal statistics~\cite{Wilczek82, myrheim} where topological defects and representation of braid groups~\cite{fredenhagen89} play the central role.  We can have, for example, the action of $2\pi$-rotation, leaving a non-trivial phase. This contrasts the $2\pi$-periodicity, essential to derive the integer spectrum for the excitation operator in our equation \eqref{number operator}. This suggests that we speak of different kinds of particle statistics due to the involvement of different symmetry groups. On the other hand, the recent work of ~\cite{Haldane08} suggests that fractal statistics can be phrased in terms of Jack polynomials, which generalize Schur polynomials (our primary tool to classify statistics). This relationship is worth looking into in the future. 

Very similar holds for many generalized statistics related to deformed canonical commutation relations. Take an example of $q$-deformations (quons) with $a_ia^{\dagger}_j-q a^{\dagger}_ja_i=\delta_{ij}\openone$~\cite{quons}.
However, $q$-deformations introduce new symmetries even at the level of a single particle, i.e., $q$-deformed $U(d)$ \cite{Jimbo85} group, while our statistics is directly paired to the $U(d)$ symmetry. Still, some comparison might be possible for the order-one statistics, where our ansatz of Section \ref{order-one statistics} provides the means to construct the algebra of creation and annihilation operators and evaluate the corresponding commutation relations. 

Finally, the question is how our generalization is related to parastatistics~\cite{Green53}. As already pointed out, the group behind the parastatistics is different~\cite{ryan1963, Stoilova_2008}. This leads to the different Fock space decomposition, i.e., for parastatistics of order $p$, we have~\cite{HARTLE69,stoilova}
\begin{align}\label{parabose}
  \mathcal{F}_{\text{parab}}&=\bigoplus_{l(\lambda)\leq p} \mathcal{V}_{\lambda}\\\label{parafermi}
  \mathcal{F}_{\text{paraf}}&=\bigoplus_{l(\lambda')\leq p} \mathcal{V}_{\lambda},
\end{align}
where the sum runs over Young diagrams $\lambda$ (parabose case) or $\lambda'$ (parafermi case) of the length $l(\lambda)$ (number of rows). Here $\lambda'$ is the conjugated diagram of $\lambda$ and $\mathcal{V}_{\lambda}$ is an $U(d)$-IR associated to $\lambda$. This decomposition contains no multiplicities and thus is compatible with our classification only for the case of ordinary statistics.

\subsection{Some open questions and applications}
The broad range of possibilities for generalized statistics introduced here leaves many interesting open questions and potential for applications. Firstly, an open question is what is more on the physical side (compared to new effects already discussed) that tanstatistics brings. As we already discussed, there are technical difficulties with higher-order statistics, mainly in the context of hidden quantum numbers. Nevertheless, we may study thermodynamics directly by the ansatz defined in section \ref{TD}. One has to calculate partition functions given in \eqref{partition function} for more general characters. In this case, a simple shift of the chemical potential in \eqref{chmical potential} will not reduce thermodynamical quantities to ones given by ordinary statistics as it happens for order-one statistics. Given this, we can expect other novel physical effects to appear. 

The next exciting point to analyze is the application of our method to diagonalize solid-state Hamiltonians, such as it is done for spin-chains via spin-fermion mapping (Jordan-Wigner transformation)~\cite{lieb61}. For example, the transfermionic Fock space for $[1,\alpha]_{+}$ is isomorphic to $(\mathbb{C}^{(\alpha+1)})^{\otimes d}$ which is suitable to study higher dimensional spin chains. In complete analogy to the spin-fermion mapping, one can expect to find other integrable many-body Hamiltonians that reduce to our non-interacting model. 

An interesting point to be analyzed is the question of entanglement in transtatistics. This question has raised a long-standing debate in the community regarding the case of bosons and fermions due to the apparent entanglement present in the first quantized picture, which comes solely from the (anti) symmetrization of the wave function. It is accepted nowadays that such ``kinematic'' entanglement is physical~\cite{benatti2020,Morris2020}. In the case of transtatistics, the starting point is different as we do not know if the first-quantized picture for such generalized statistics exists; thus, the situation is less clear. However, one can try to put transtatistics in the context of quantum-information processing and protocols designed to study the entanglement of standard indistinguishable particles (see ~\cite{Morris2020} and references therein). This shall provide a more clear view of the relation between entanglement and indistinguishability in this case. Along these lines, an exciting perspective on our results comes from the quantum computational complexity of quantum statistics. Namely, it is well-known that the non-interacting bosons are computationally hard to simulate~\cite{aaronson11}, while non-interacting fermions are not~\cite{terhal02}. One can ask a similar question here, i.e., what is the computational power of the non-interacting model for transtatistics? Any answer to it is relevant and may find applications in quantum computing. 

Finally, on the speculative side, an interesting idea of applying generalized statistics in the context of dark-matter modeling was recently presented~\cite{hoyuelos2022dark}. The main point is to study thermodynamics and the effects of the negative relation between pressure and energy density, emphasized in many existent dark energy candidates. Our methods provide a direct way to calculate thermodynamical properties of transtatistics and thus might be worthy of investigating relations to dark-matter models.

\begin{acknowledgements}
The authors thank S. Horvat, J. Morris and {\v C}. Brukner for their helpful comments. This research was funded in whole, or in part, by the Austrian Science Fund (FWF) [10.55776/F71], [10.55776/P36994] and [10.55776/COE1] and the European Union – NextGenerationEU. For open access purposes, the author(s) has applied a CC BY public copyright license to any author accepted manuscript version arising from this submission.
\end{acknowledgements}

\bibliography{apssamp}

\onecolumngrid
\appendix

 \section{Schur polynomials}\label{Schur polynomails}
The linear representations of the general linear group $GL(n,\mathbb{C})$ and its maximally compact subgroup $U(n)$ are unambiguously identified by \emph{Schur polynomials} as their characters~\cite{fulton}. These families of polynomials appear in different regions of mathematics, from pure combinatorics to algebraic geometry. This is why several standard definitions exist depending on the specific context they are discussed. Here we present them in the combinatorial definition to emphasize the combinatorial of our operational reconstruction of particle statistics. Other methods, such as the classical (determinant) definition, are standardly found in the literature ~\cite{stanley_fomin_1999}. 

Let $\lambda = (\lambda_1, \lambda_2, \dots)$ be an integer partition with $\lambda_1 \geq \lambda_2 \geq \dots$, usually represented with a \emph{Young diagram} (see Figure \ref{fig:my_label}). The total number of boxes in a diagram is denoted by $|\lambda| = \lambda_1 + \lambda_2 + \dots$, and the partition length (the number of rows) is labeled by $l(\lambda)$. 

 \begin{figure}[h!]
     \centering
     \includegraphics[scale=0.5]{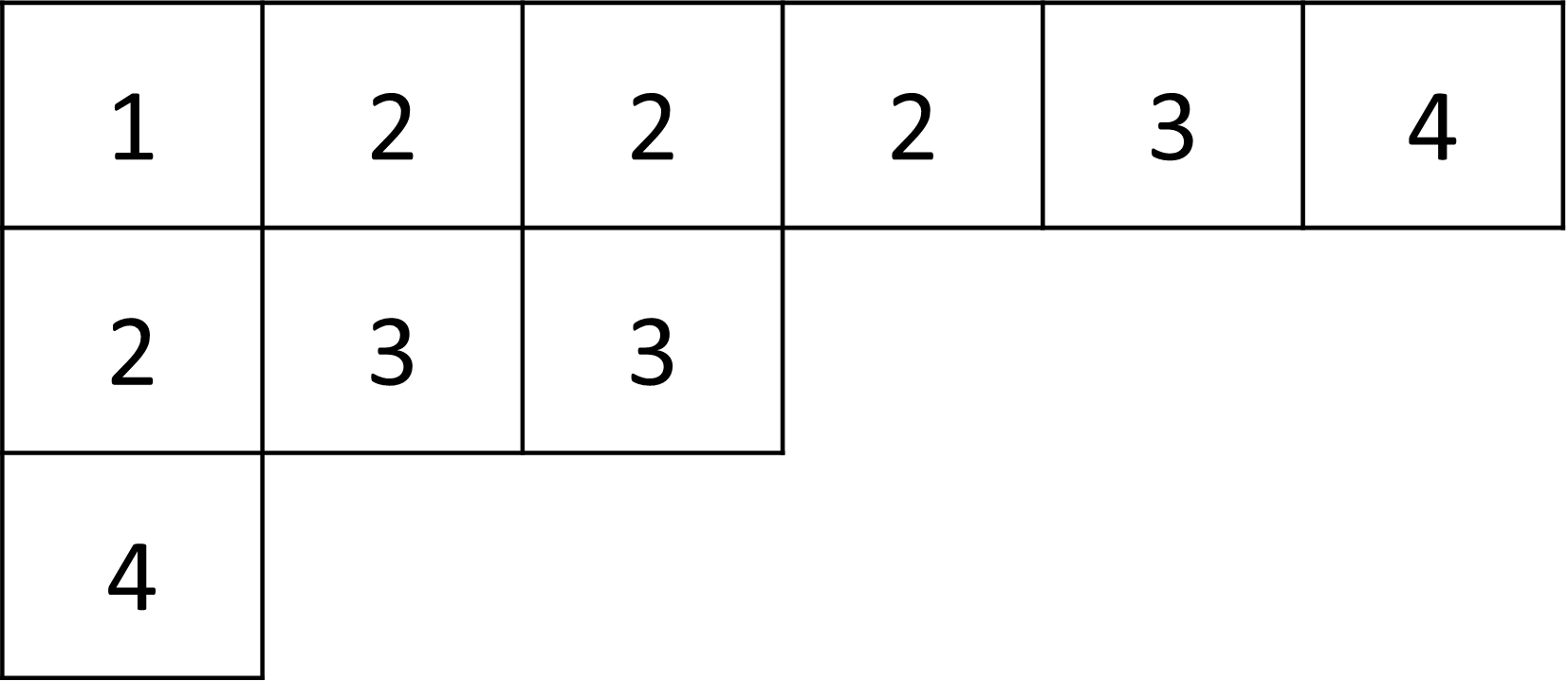}
     \caption{\hbox{Young diagram for $\lambda=(6,3,1)$}. The filling of the diagram defines one SSYT.}
     \label{fig:my_label}
 \end{figure}
 The \emph{semistandard} Young tableau (SSYT) of shape $\lambda$ is a \emph{filling} of the boxes in the Young diagram with positive integers such that the entries weakly increase along each row and strictly increase down each column. For given a SSYT $T$, we define the \emph{type} of $T$,  i.e., $\alpha(T)=(\alpha_1,\alpha_2,\dots)$ with $\alpha_i(T)$ being the number of repetitions of the number $i$ in $T$. For example, for the SSYT given in Fig. \ref{fig:my_label}, we have $\alpha(T)=(1,4,3,2)$. For a set of variables $x_1, x_2, \dots$ we define
 \begin{equation}
x^{T}\equiv x_1^{\alpha_1(T)}x_2^{\alpha_2(T)}\dots.
 \end{equation}
\begin{definition}
Let $\lambda$ be a partition. The \emph{Schur function} $s_\lambda(x_1,\dots,x_d)$ in $d\geq l(\lambda)$ variables associated with $\lambda$ is a homogeneous symmetric polynomial of degree $|\lambda|$ defined as:
\begin{equation}
s_\lambda(x)=\sum_{T\in\mathrm{SSYT}_\lambda(d)} x^T,
\end{equation}
where the sum runs over all semistandard Young tableaux (SSYTs) of shape $\lambda$ with the filling from the set $\{1,\dots,d\}$.
\end{definition}
The common examples are the bosonic Schur functions (\emph{homogeneous symmetric polynomials}) 
\begin{equation}
    s_{(n,0,\dots,0)}(x_1,\dots,x_d)=\sum_{k_1+\dots+k_d=n}x_1^{k_1}\dots x_d^{k_d},~~~~k_s=0\dots d,
\end{equation}
associated with partitions having a single row of size $n$ ($n$-particle sector), and the fermionic Schur functions (\emph{elementary symmetric polynomials}) 
\begin{equation}
  s_{(1,\dots,1,0,\dots,0)}(x_1,\dots,x_d)=\sum_{1\leq i_1<i_2\dots< i_n\leq d}x_{i_1}x_{i_2}\dots x_{i_n},  
\end{equation}
associated with partitions having a single column of size $n$ ($n$-particle sector). Notably, these bosonic and fermionic partitions are conjugate to each other. Two partitions are said to be conjugate if they can be obtained from each other by interchanging rows and columns in their Young diagrams.
 \section{Proof of main (partition) theorem}\label{main proof}
Before proceeding with the proof, we must introduce some basic definitions and important known theorems. 
\begin{definition}
 We call $f(x)=\sum_n a_n x^n$ the generating function of sequence $\lbrace a_n \rbrace$. 
\end{definition}
\begin{definition}
A matrix $(A_{ij})_{i,j\in I}$ is called T{\" o}plitz if $A_{i,j}=a_{i-j}$.
\end{definition}
\begin{definition}
A sequence of real numbers $\lbrace a_n \rbrace_{n\in \mathbb{Z}}$ is totally positive if and only if all the minors of the Töplitz matrix $(a_{i-j})_{i,j\in \mathbb{Z}}$ are non-negative. For sequences defined only on non-negative integers $\lbrace a_n \rbrace_{n\in \mathbb{N}_0}$, we assume the extension $a_n=0$ for $n<0$. Generating function $f(x)=\sum_n a_n x^n$ is called totally positive if and only if $\{a_n\}$ is totally positive.
\end{definition}
\begin{proposition}[\cite{bump}]\label{Proposition Bumb}
Let $f(x)=\sum_n a_n x^n$ and $D_{d-1}^{\lambda,\mu}(f)=\det (a_{\lambda_i-\mu_j-i+j})_{1\leq i,j\leq d}$ with $\lambda$ and $\mu$ being a pair of partitions of possibly different integers. Then every T{\" o}plitz minor of matrix $(a_{i-j})_{i,j\in \mathbb{Z}}$ is of the form $D_{d-1}^{\lambda,\mu}(f)$. Furthermore, we have
\begin{equation}
D_{d-1}^{\lambda,\mu}(f)=\int_{U(d)}\Phi_{n,f}(g)\overline{s_\lambda(g)}s_\mu(g)dg
\end{equation}
with $\Phi_{d,f}=f(x_1)...f(x_d)$ where $x_1,...,x_d$ are the eigenvalues of $g$, and $\int_{U(d)}\dots dg$ is the Haar measure integral.
\end{proposition}
\begin{proposition}[\cite{Schoenberg}]\label{ASW theorem}
The sequence $\lbrace a_n \rbrace_{n\in\mathbb{N}_0}$ with $a_0=1$ is totally positive if and only if it is generated by a function $f(x)$ of the form
\begin{equation}
f(x)=e^{\gamma x}\frac{\Pi_i (1+\alpha_i x)}{\Pi_i (1-\beta_i x)}
\end{equation}
with $\alpha,\beta,\gamma \geq 0$ and $\sum \alpha_i,\sum \beta_i$ convergent.
\end{proposition}
The last proposition is the most prominent result for characterizing totally-positive sequences. We can prove now the following simple lemma (an almost equivalent statement was presented in \cite{integer, Davydov}).
\begin{lemma}\label{integral lemma}
An integral series $f(x)=\sum_{n=0}^\infty a_n x^n$ with $a_0>0$ is totally positive if and only if it is of the form
\begin{equation}
f(x)=\frac{g(x)}{h(x)},
\end{equation}
for some integral polynomials $g(x)$ and $h(x)$ with $h(0)=1$, such that all complex roots of $g(x)$ are negative and all those of $h(x)$ are positive real numbers.    
\end{lemma}
\begin{proof}
A sequence $\{a_n\}_{\mathbb{N}_0}$ is totally positive if and only if $\{r a_n\}_{\mathbb{N}_0}$ is totally positive for $r>0$. We set $r=1/a_0$ and by the Proposition \ref{ASW theorem}, we get that 

\begin{equation}
f(x)=a_0e^{\gamma x}\frac{\Pi_i (1+\alpha_i x)}{\Pi_i (1-\beta_i x)}
\end{equation}
Because $\sum_i(\alpha_i+\beta_i)<+\infty$, the function $f$ is meromorphic in $|x|\leq1$ with a finite number of poles and zeros inside the unit disc. By the theorem of Salem~\cite{Salem}, this function is rational, i.e., of the form $g(x)/h(x)$. It is not difficult to show that for rational function $g(x)/h(x)=\sum_{n\in\mathbb{N}_0}a_nx^n$ with $a_n$ integers, polynomials $g(x)$ and $h(x)$ are relatively prime with $h(0)=1$ (see exercise $2(a)$ in Chapter 4 of ~\cite{stanley_2011}).    
\end{proof}
Having these in mind, we can write down the proof of the Partition theorem \ref{Partition theorem}. For clarity, we repeat the statement.
\begin{thm1}[Partition]
For $\chi_1(x)=\sum_{s\in\mathbb{N}_0}a_sx^s$ with $a_0>0$, a symmetric function $\prod_{k=1}^d \chi_1(x_k)$ is a $U(d)$ character for all $d\in\mathbb{N}$ if and only if the generating function is of the form
\begin{equation}
\chi_1(x)=\frac{Q_{-}(x)}{Q_+(x)},
 \end{equation}
where $Q_{\pm}(x)$ is an integral polynomial with all positive (negative) roots.  Furthermore $Q_+(0)=1$.
\end{thm1}
\begin{proof}\label{proof of Partition theorem}
Clearly $\{a_n\}$ is integral because $\chi_1(x)$ is a character of $U(1)$. The character $\chi_d(x_1,...,x_d)=\prod_{k=1}^d \chi_1(x_k)$ is a class function over $U(d)$ and $x_1,...,x_d$ are variables on the maximal torus in $U(d)$, i.e., $x_k=e^{i\theta_k}$. For any symmetric function $f(x_1,...,x_d)$, the notation $f(g)$ assumes $f$ being evaluated at eigenvalues $x_1,\dots,x_d$ of matrix $g\in U(d)$. Now, as a consequence of the Littlewood-Richardson rule, the symmetric function $\chi_d(x_1,...,x_d)$ is Schur-positive (expands in non-negative coefficients over Schur polynomials) if and only if the product $s_\mu (x_1,...,x_d)\chi_d(x_1,..,x_d)$ is also Schur-positive. Using this, we expand $s_\mu(x_1,...,x_d)\chi(x_1,...,x_d)$ in Schur polynomials as
\begin{equation}
s_\mu(x_1,...,x_d)\chi_d(x_1,...,x_d)=\sum_\lambda c_{\mu\lambda} s_\lambda(x_1,...,x_d)
\end{equation}
where the sum runs over Young diagrams $\lambda$ and $c_{\mu\lambda}\geq0$. Since Schur polynomials are orthogonal under the Haar measure 

\begin{equation}\label{B6}
\langle s_\lambda,s_\mu \rangle=\int_{U(d)}\overline{s_\lambda(g)}s_\mu(g)dg=\delta_{\lambda,\mu},
\end{equation}
we get the expression 
\begin{equation}
    c_{\mu\lambda}=\int_{U(d)}\chi_d(g)\overline{s_\lambda}(g)s_\mu(g)dg
\end{equation}
Using the fact that character is of factorization form $\chi_d(x_1,\dots,x_d)=\prod_{k=1}^d\chi_1(x_k)$, the conditions of the Proposition \ref{Proposition Bumb} are met and we can rewrite the previous expression as
\begin{equation}
c_{\mu\lambda}=\det(a_{\lambda_i-\mu_j-i+j})_{1\leq i,j\leq d}=D_{d-1}^{\lambda,\mu}(\chi_1),
\end{equation}
where $\{a_n\}$ is a sequence that generates the single-mode character $\chi_1(x)$. Schur-positivity $c_{\mu\lambda}\geq0$ thus reads 
\begin{equation}\label{B10}
D_{d-1}^{\lambda,\mu}(\chi_1)=\det(a_{\lambda_i-\mu_j-i+j})_{1\leq i,j\leq d}\geq 0,
\end{equation}
which is by Proposition \ref{Proposition Bumb} equivalent to the condition that all minors of the T{\"o}plitz matrix $(a_{i-j})_{i,j\in \mathbb{Z}}$ are non-negative. This means $\{a_n\}$ is totally positive. $\{a_n\}$ is also an integral sequence with $a_0>0$, therefore the main result follows by Lemma \ref{integral lemma}. 
\end{proof}
\section{Fock space decomposition for order-one statistics}\label{proof of Fock space dec}
Suppose two sets of variables $x_1,...,x_n$ and $y_1,...,y_m$ with $m\leq n$. The Cauchy identities \cite{stanley_fomin_1999} are the following 
\begin{eqnarray}
    \prod_{i=1}^n\prod_{j=1}^m (1+x_iy_j)&=&\Sigma_{l(\lambda) \leq m} s_{\lambda}(x_1,...,x_n)s_{\lambda}(y_1,...,y_m),\\
    \prod_{i=1}^n\prod_{j=1}^m\frac{1}{1-x_iy_j}&=&\Sigma_{l(\lambda) \leq m} s_\lambda(x_1,...,x_n)s_\lambda(y_1,...,y_m)
\end{eqnarray}
where $l(\lambda)$ is the length of the diagram (number of rows), and $\lambda'$ is the conjugate partition of $\lambda$. Using these identities, the proof of Theorem \ref{Polarity} is as follows.
  \begin{proof}
  For statistics $[1,\alpha]_{-}$ and $[1,\beta]_{+}$ we have the following $d$-mode characters (defined in \eqref{character})
  \begin{eqnarray}
      \chi^{-}_d(\vec{x})&=&\prod_{i=1}^d(1+\alpha x_i),\\
      \chi^{+}_d(\vec{x})&=&\prod_{i=1}^d\frac{1}{1-\beta x_i}.
  \end{eqnarray}
Now, we apply the first Cauchy identity by setting $m=1$ and $\alpha=y_1$ and we get
  \begin{align}
      \chi^{-}_d(\vec{x})&=\Sigma_{l(\lambda) \leq 1}s_{\lambda}(\alpha)s_{\lambda'}(x_1,...,x_d)\\
      &=\Sigma_{l(\lambda) \leq 1}\alpha^{|\lambda|}s_{\lambda'}(x_1,...,x_d).
     \end{align}
Similarly, we apply the second Cauchy identity by setting $m=1$ and $\beta=y_1$ and we get
     \begin{align}
      \chi^{+}_d(\vec{x})&=\Sigma_{l(\lambda) \leq 1}s_{\lambda}(\beta)s_\lambda(x_1,...,x_d)\\
      &=\Sigma_{l(\lambda) \leq 1}\beta^{|\lambda|}s_\lambda(x_1,...,x_d).
  \end{align}
  \end{proof}
\section{Negative occupations and general character}\label{General character}
We set the most general $U(1)$ character as a formal Laurent series  
\begin{equation}\label{Laurent character}
    \chi_1(x)=\sum_{s\in\mathbb{Z}}a_s x^s,
\end{equation}
with integral coefficients, i.e. $a_s\in\mathbb{N}_0$. For the set of $d$ modes, the character decomposes as
\begin{equation}\label{general character decomposition}
  \chi_d(x_1,\dots,x_d)=\prod_{s=1}^d\chi_1(x_s)=
  \sum_{k,\lambda} c_{\lambda}^{(k)} s_\lambda^{(k)}(\vec{x}),  
\end{equation}
with $c_{\lambda}^{(k)}\in\mathbb{N}_0$ and irreducible factors being Laurent-Schur polynomials of the form
\begin{equation}\label{Schur-Laurent}
    s_\lambda^{(k)}(x_1,\dots,x_d)=(x_1\dots x_d)^{-k}s_{\lambda}(x_1,\dots,x_d),~~~k\in\mathbb{N}_0.
\end{equation}
These are associated with the (rational) irreducible representations of $U(d)$~\cite{stanley_fomin_1999}. The factor $(x_1\dots x_d)^{-k}=(\mathrm{det}g)^{-k}$ corresponds to the determinant representation, while $s_{\lambda}(x_1,\dots,x_d)$ is the standard Schur polynomial.

Before proceeding further, a few comments are needed. Firstly, we adopt the notation of $k^d=(k,k,\dots,k)$ being the partition of length $d$, and we define $\lambda+k^d=(l_1+k,\dots,\lambda_d+k)$. Now, note that indexing of Laurent-Schur polynomials in \eqref{Schur-Laurent} with a pair $(k,\lambda)$ is not unique, i.e., $s_\lambda^{(k)}(\vec{x})=s_\mu^{(l)}(\vec{x})$ holds whenever $\lambda+l^d=\mu+k^d$. This is because $s_{\lambda+k^d}(\vec{x})=(x_1\dots x_d)^{k}s_\lambda(\vec{x})$~\cite{bump2004lie}. In the expansion \eqref{general character decomposition} we implicitly assume summation over different $s_\lambda^{(k)}$s only, although we write explicitly summation over pairs $(k,\lambda)$ to simplify the notation. Next, functions in \eqref{Schur-Laurent} are orthogonal under Haar measure, i.e.,
\begin{equation}
\langle s_\lambda^{(k)},s_\mu^{(l)} \rangle=\int_{U(d)}(\mathrm{det}g)^{k}\overline{s_{\lambda}(g)}(\mathrm{det}g)^{-l}s_{\mu}(g)dg=\int_{U(d)}\overline{s_{\lambda+l^d}(g)}s_{\mu+k^d}(g)dg=\delta_{\lambda+l^d,\mu+k^d},
\end{equation}
With this, the proof of the main theorem \ref{proof of Partition theorem} and equations \eqref{B6}-\eqref{B10} trivially modify by changing Schur polynomials $s_\lambda$ with $s_\lambda^{(k)}$, and we easily conclude that the generating sequence of \eqref{Laurent character} has to be totally-positive. Edrei has provided their full classification: 
\begin{proposition}[\cite{edrei1953generation}]\label{Edrei theorem}
Let $\lbrace a_s \rbrace_{s\in\mathbb{Z}}$ be a totally-positive sequence and the Laurent series
\begin{equation}\label{Laurent series}
   \sum_{s\in\mathbb{Z}}a_s x^s
\end{equation}
associated with it. If the sequence does not
coincide with a sequence of the form
\begin{equation}\label{exceptional sequence}
    \lbrace a \rho^s \rbrace_{s\in\mathbb{Z}}\qquad\qquad\qquad\qquad(a,\rho>0), 
\end{equation}
then the Laurent series \eqref{Laurent series} converges in some ring
\begin{equation}\label{convergence in ring}
    r_1<|x|< r_2\qquad\qquad\qquad\qquad(0\leq r_1<r_2),
\end{equation}
and the analytic continuation of \eqref{Laurent series} is of the form
\begin{equation}
f(x)=Cx^ke^{c_1 x+c_{-1}x^{-1}}\frac{\Pi_{s\in\mathbb{N}} (1+\alpha_s x)}{\Pi_{s\in\mathbb{N}} (1-\beta_s x)}\frac{\Pi_{s\in\mathbb{N}} (1+\gamma_s x^{-1})}{\Pi_{s\in\mathbb{N}} (1-\delta_s x^{-1})},
\end{equation}
with $k\in\mathbb{Z}$, $C,c_1,c_{-1},\alpha,\beta,\gamma,\delta \geq 0$, and $\sum_{s\in\mathbb{N}} \alpha_s+\beta_s+\gamma_s+\delta_s<+\infty$.
\end{proposition}
Now, we are ready to prove the following statement:
\begin{proposition}\label{finiteness of character}
Let $\lbrace a_s \rbrace_{s\in\mathbb{Z}}$ be an integral $(a_s\in\mathbb{N}_0)$ totally-positive sequence which is not of the form \eqref{exceptional sequence}. Then this sequence cannot extend to double infinity, i.e., there exists $k\in\mathbb{Z}$ such that $a_k=0$ either for $s<k$ or $s>k$.
\end{proposition}
\begin{proof}
    Suppose by contradiction that $a_s$ is non-zero at both $\pm\infty$. We divide the associated Laurent series \eqref{Laurent series} into the positive $S_+(x)=\sum_{s=1}^{+\infty}a_s x^s$ and negative $S_-(x)=\sum_{s=-\infty}^{0}a_s x^s$ parts. Because $a_s$ is non-zero (integer) at infinity, the series $S_+$ converges at best within the unit radius, i.e., $|x|<r\leq1$. Similarly, $S_-$ converges within the radius $1/|x|<R\leq1$. These two convergence conditions are compatible only for $|x|=1$. Since $\lim_{n\to\pm \infty}a_{s}x^s\neq0$ (or does not exist) for $|x|=1$, we conclude \eqref{Laurent series} converges nowhere. This contradicts Proposition \ref{Edrei theorem} because of \eqref{convergence in ring}.  
\end{proof}
From the last proposition, it follows that the general single-mode character \eqref{Laurent character} is either of the form $\chi_1(x)=x^{k}\sum_{s\geq0}c_s x^s$ or $\chi_1(x)=x^{k}(\sum_{s\geq0}c_s x^s)^\ast$ for some $k\in\mathbb{Z}$, where $\ast$ denotes complex conjugate (note that $x^{\ast}=1/x$ because $x=e^{i\varphi}$). The factor $x^k$ builds into $(x_1\dots x_d)^k$ for $d$ modes and corresponds to the determinant representations (power $k$). Our main theorem \ref{Partition theorem} applies directly to the factor $\sum_{s\geq0}c_s x^s$. As follows from Proposition \ref{Edrei theorem}, the only exceptional case is the character generated by the sequence \eqref{exceptional sequence} with $a,\rho\in\mathbb{N}_0$.   
\end{document}